\documentclass[sigconf, nonacm]{acmart}

\newcommand{\nop}[1]{}

\newcommand\vldbdoi{XX.XX/XXX.XX}
\newcommand\vldbpages{XXX-XXX}
\newcommand\vldbvolume{14}
\newcommand\vldbissue{1}
\newcommand\vldbyear{2020}
\newcommand\vldbauthors{\authors}
\newcommand\vldbtitle{\shorttitle} 
\newcommand\vldbavailabilityurl{URL_TO_YOUR_ARTIFACTS}
\newcommand\vldbpagestyle{plain} 
\newtheorem{definition}{Definition}
\newtheorem{example}{Example}
\newtheorem{lemma}{Lemma}
\usepackage[linesnumbered,ruled,vlined]{algorithm2e}
\usepackage{amsthm}
\usepackage{bm}
\usepackage{float}
 \usepackage{balance}
\usepackage{enumitem}
\usepackage{tcolorbox}
\usepackage{subfigure}

\begin{document}

\title{S$^3$AND: Efficient Subgraph Similarity Search Under Aggregated Neighbor Difference Semantics (Technical Report)}

\author{Qi Wen}
\affiliation{%
  \institution{East China Normal University}
  \city{Shanghai}
  \country{China}
}
\email{51265902057@stu.ecnu.edu.cn}
\author{Yutong Ye}
\affiliation{%
  \institution{East China Normal University}
  \city{Shanghai}
  \country{China}
}
\email{52205902007@stu.ecnu.edu.cn}

\author{Xiang Lian}
\affiliation{%
  \institution{Kent State University}
  \city{Kent}
  \state{Ohio}
  \country{USA}
}
\email{xlian@kent.edu}

\author{Mingsong Chen}
\affiliation{%
  \institution{East China Normal University}
  \city{Shanghai}
  \country{China}
}
\email{mschen@sei.ecnu.edu.cn}








\begin{abstract}

For the past decades, the \textit{subgraph similarity search} over a large-scale data graph has become increasingly important and crucial in many real-world applications, such as social network analysis, bioinformatics network analytics, knowledge graph discovery, and many others. While previous works on subgraph similarity search used various graph similarity metrics such as the graph isomorphism, graph edit distance, and so on, in this paper, we propose a novel problem, namely \textit{subgraph similarity search under aggregated neighbor difference semantics} (S$^3$AND), which identifies subgraphs $g$ in a data graph $G$ that are similar to a given query graph $q$ by considering both keywords and graph structures (under new keyword/structural matching semantics). To efficiently tackle the S$^3$AND problem, we design two effective pruning methods, \textit{keyword set} and \textit{aggregated neighbor difference lower bound pruning}, which rule out false alarms of candidate vertices/subgraphs to reduce the S$^3$AND search space. Furthermore, we construct an effective indexing mechanism to facilitate our proposed efficient S$^3$AND query answering algorithm. Through extensive experiments, we demonstrate the effectiveness and efficiency of our S$^3$AND approach over both real and synthetic graphs under various parameter settings.

\end{abstract}

\maketitle

\pagestyle{\vldbpagestyle}
\begingroup\small\noindent\raggedright\textbf{PVLDB Reference Format:}\\
\vldbauthors. \vldbtitle. PVLDB, \vldbvolume(\vldbissue): \vldbpages, \vldbyear.\\
\href{https://doi.org/\vldbdoi}{doi:\vldbdoi}
\endgroup
\begingroup
\renewcommand\thefootnote{}\footnote{\noindent
This work is licensed under the Creative Commons BY-NC-ND 4.0 International License. Visit \url{https://creativecommons.org/licenses/by-nc-nd/4.0/} to view a copy of this license. For any use beyond those covered by this license, obtain permission by emailing \href{mailto:info@vldb.org}{info@vldb.org}. Copyright is held by the owner/author(s). Publication rights licensed to the VLDB Endowment. \\
\raggedright Proceedings of the VLDB Endowment, Vol. \vldbvolume, No. \vldbissue\ %
ISSN 2150-8097. \\
\href{https://doi.org/\vldbdoi}{doi:\vldbdoi} \\
}\addtocounter{footnote}{-1}\endgroup

\ifdefempty{\vldbavailabilityurl}{}{
\vspace{.3cm}
\begingroup\small\noindent\raggedright\textbf{PVLDB Artifact Availability:}\\
The source code, data, and/or other artifacts have been made available at \url{https://github.com/Luminous-wq/S3AND}.
\endgroup
}


\section{Introduction}
\label{sec:intro}

The \textit{subgraph similarity search} over graphs has been widely used as an important and fundamental tool for real-world applications, such as social network analysis~\cite{rai2023top}, knowledge graph discovery~\cite{song2018mining}, bioinformatics mining~\cite{GliozzoPGCV23}, and so on. Specifically, a subgraph similarity search query retrieves those subgraphs $g$ in a large-scale data graph $G$ that are similar to a given query graph pattern $q$.


Existing works on the subgraph similarity search used graph similarity metrics, such as \textit{graph edit distance}~\cite{gao2010survey,blumenthal2020comparing} and \textit{chi-square statistics}~\cite{dutta2017neighbor}, to measure the similarity between subgraphs $g$ and query graph $q$. While different graph similarity semantics are helpful for different real applications (e.g., with similar graph structures or statistics), in this paper, we propose a novel graph similarity measure, called \textit{aggregated neighbor difference} (AND), which is given by aggregating the neighbor differences of the matching vertices between subgraph $g$ and query graph $q$. Based on this AND semantic, we formulate a new problem, namely \textit{subgraph similarity search under aggregated neighbor difference semantics} (S$^3$AND), which obtains subgraphs $g \subseteq G$ that match with $q$ with low AND scores.




Below, we give a motivation example of our S$^3$AND problem in the application of collaboration social network analysis.

\setlength{\textfloatsep}{2pt}
\begin{figure}[t]
    \centering\vspace{2ex}
    \subfigure[collaboration social network $G$] {
        \scalebox{0.26}[0.26]{\includegraphics{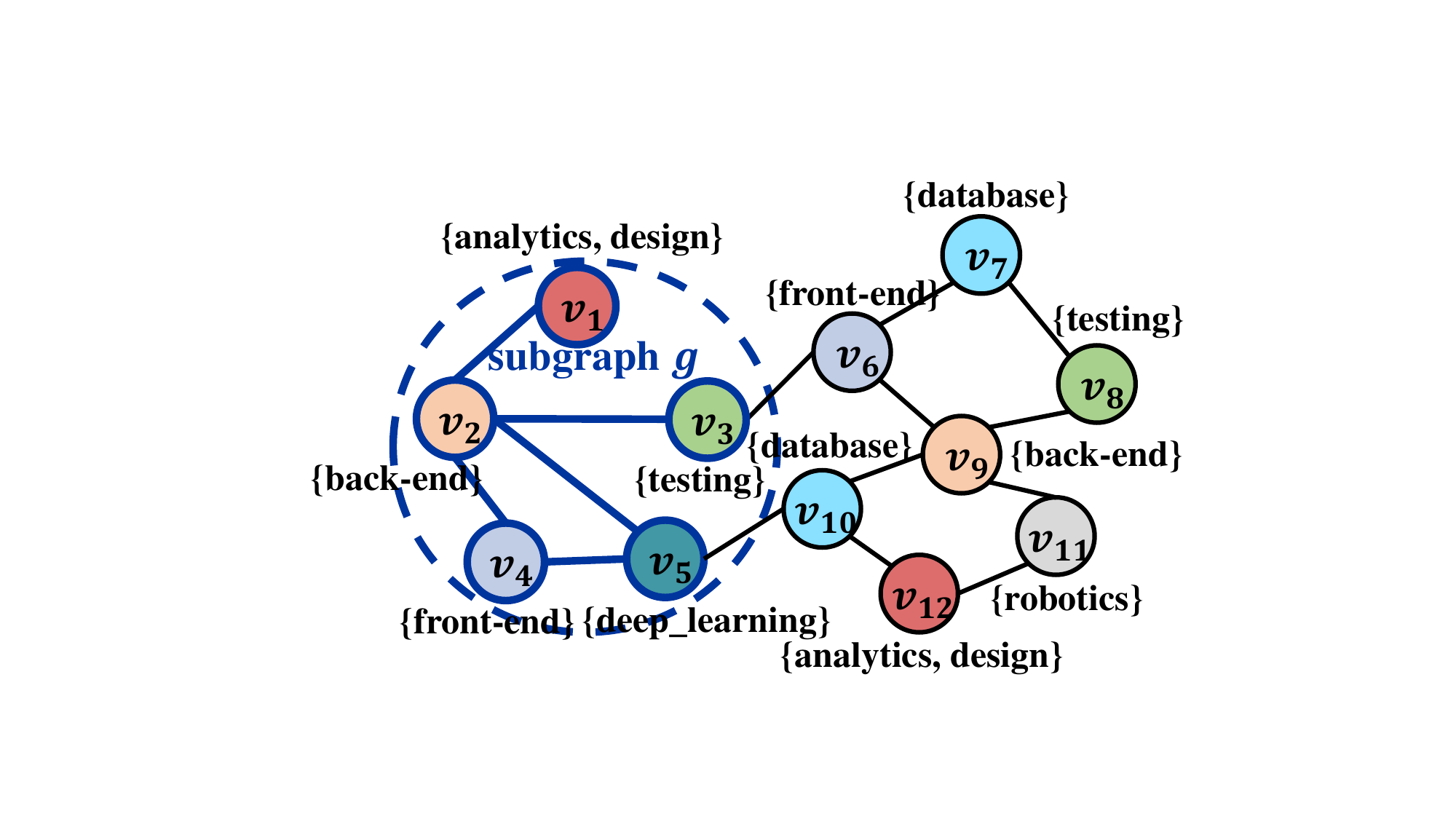}}
        \label{subfig:example_a}
    }\hspace{-1.9ex}%
    \subfigure[target team (query graph) $q$]{
    \scalebox{0.298}[0.298]{\includegraphics{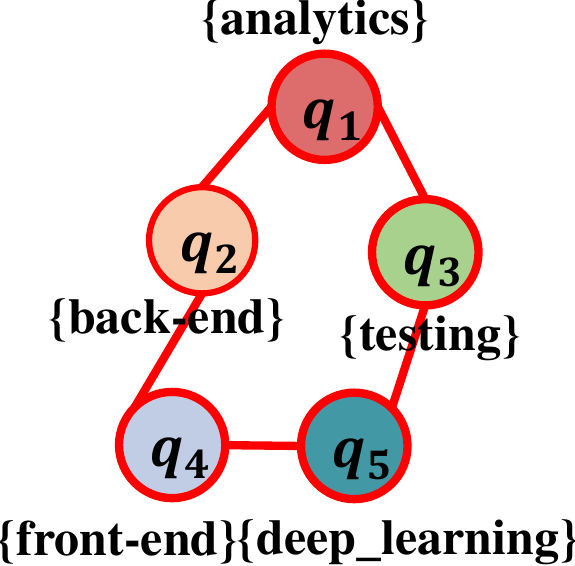}}
        \label{subfig:example_b}
    }
    \vspace{-2ex}
    \caption{An S$^3$AND example of the skilled team search.}
    \label{fig:AND_application}\vspace{1ex}
\end{figure}

\begin{example}
\label{example1}
\textbf{(The Skilled Team Search in Collaboration Social Network)} 
To accomplish a new project, a manager wants to recruit an experienced team that consists of members with relevant skills and previous collaboration experiences. Figure~\ref{subfig:example_a} shows a collaboration social network $G$, which consists of 12 user vertices, $v_1 \sim v_{12}$, each with a set of skill labels (e.g., user $v_2$ has the ``back-end'' development skills), and collaborative edges (each connecting two users, e.g., $v_2$ and $v_3$, indicating that they collaborated on some project before). 

Figure ~\ref{subfig:example_b} shows a target (query) graph pattern $q$, which represents a desirable team structure, specified by the project manager. In particular, each member $q_i$ ($1\leq i \leq 5$) in this experienced team must have certain skills (i.e., query keywords), for example, team member $q_2$ should have the ``back-end'' skill. Moreover, during the project period, team members are required to communicate frequently for accomplishing the project together. Thus, it is preferred that they have previous experience in project collaboration to reduce the one-to-one communication overhead. As an example, the ``back-end'' team member $q_2$ is expected to have collaborative experience (i.e., edge $e(q_2, q_4)$)  with a ``front-end'' member $q_4$ before. 



However, in practice, it is rare to find a perfect subgraph that exactly matches the query graph $q$. For example, in Figure ~\ref{fig:AND_application}, we cannot find a subgraph of $G$ that is structurally isomorphic to the query graph $q$. 

Alternatively, the manager can issue an S$^3$AND query to obtain a team from $G$ (e.g., the subgraph $g$ within the dashed circle of Figure~\ref{subfig:example_a}), whose members have the required skills to accomplish project tasks (i.e., data vertices in $g$ must contain the required keywords in query vertices), but follow some relaxed constraints on the collaboration experience. As an example, in subgraph $g$ of Figure~\ref{subfig:example_a}, although user $v_1$ does not have experience working with user $v_3$ before (as required by edge $e(q_1, q_3)$ in query graph $q$), they can still build good collaborative relationships through the new project, however, with some extra costs (e.g., time delays and/or communication efforts, as they were not familiar with each other before). Similarly, compared with the target team $q$, collaborative edge $e(v_3, v_5)$ is also missing in $g$. Thus, although subgraph $g$ and query graph $q$ do not structurally match with each other, subgraph $g$ can still be a potential candidate team that follows strict skill constraints and meets the relaxed collaboration requirements (e.g., within some budget of extra collaboration costs, defined later as the aggregated neighbor differences (AND) in Section \ref{subsec:AND}). The S$^3$AND query can exactly help obtain such a team (subgraph) $g$ in $G$, satisfying the keyword set containment relationship between query and data vertices and with low collaboration/communication overheads (i.e., AND scores). \qquad $\blacksquare$

\end{example}



{\color{red}
}

The S$^3$AND problem has many other real applications. For example, in the Semantic Web applications~\cite{zheng2016semantic}, a SPARQL query can be considered as a query graph $q$ over a large knowledge graph $G$. Our S$^3$AND query can be used to return RDF subgraphs that follow the keyword constraints and have minor structural changes compared with $q$. 



Inspired by the examples above, our S$^3$AND problem considers novel \textit{aggregated neighbor difference} (AND) semantics for subgraph similarity search over a large data graph $G$. Efficient and effective S$^3$AND query answering is quite challenging, due to complex graph manipulations (e.g., graph structure/keyword checking and AND score calculation over a large-scale data graph). Therefore, in this paper, we will design a general framework for S$^3$AND query processing, which seamlessly integrates our proposed effective pruning strategies (with respect to keywords and AND scores) to reduce the problem search space, effective indexing mechanisms over pre-computed data from graph $G$, and efficient S$^3$AND query algorithm via the index traversal. 


Specifically, we make the following contributions in this paper.


\begin{enumerate}
    \item We formulate a novel problem, \textit{subgraph similarity search under aggregated neighbor difference semantics} (S$^3$AND) in Section \ref{sec:problem_def}, which is useful for real application scenarios.
    \item We propose a general framework for tackling our S$^3$AND problem efficiently and effectively in Section \ref{sec:framework}.
    \item We design two effective pruning strategies (w.r.t. constraints of keywords and aggregated neighbor differences) in Section \ref{sec:pruning} to filter out false alarms of candidate vertices/subgraphs and reduce the S$^3$AND search space.
    \item We devise an effective indexing mechanism to facilitate our proposed query algorithm for efficiently retrieving S$^3$AND query answers in Section \ref{sec:S3AND_query_processing}.
    \item We validate the effectiveness of our proposed pruning strategies and the efficiency of the S$^3$AND algorithm in Section \ref{sec:exper} through extensive experiments on real/synthetic graphs.
\end{enumerate}
Section \ref{sec:related_work} overviews previous works on subgraph matching and subgraph similarity search. Finally, Section \ref{sec:conclusions} concludes this paper.

\section{Problem Definition}

\label{sec:problem_def}

In this section, we give the definitions of the graph data model, neighbor difference semantics, and the \textit{subgraph similarity search under aggregated neighbor difference semantics} (S$^3$AND) problem.

\subsection{Graph Data Model}

We first provide the formal definition of a large-scale data graph $G$. 

\begin{definition}
    \textbf{(Data Graph, \bm{$G$})} A data graph $G$ is in the form of a triple $(V(G), E(G), \Phi(G))$, where $V(G)$ is a set of vertices, $v_i$, in graph $G$, each with a keyword set $v_i.W$, $E(G)$ represents a set of edges $e(v_i, v_j)$ (connecting two ending vertices $v_i$ and $v_j$), and $\Phi(G)$ is a mapping function: $V(G) \times V(G) \rightarrow E(G)$. 
    \label{def:data_graph}
\end{definition}



Examples of the data graph in Definition \ref{def:data_graph} include social networks \cite{Al-BaghdadiL20, rai2023top, subramani2023gradient}, bioinformatics networks \cite{GliozzoPGCV23}, financial transaction networks \cite{SongZK23}, and so on.

\subsection{Aggregated Neighbor Difference Semantics}
\label{subsec:AND}

\noindent {\bf The Vertex-to-Vertex Mapping, $M: V(q)\to V(g)$:} Consider a target (query) graph pattern $q$ and a subgraph $g$ of data graph $G$ with the same graph size, that is, $|V(g)| = |V(q)|$. We say that there is a vertex-to-vertex mapping, $M: V(q)\to V(g)$, between $q$ and $g$, if each vertex $q_j \in V(g)$ has a 1-to-1 mapping to a query vertex $q_j\in V(q)$, such that their keyword sets satisfy the condition that $q_j.W \subseteq v_i.W$.

\noindent {\bf The Vertex Subset Mapping Function, $\mu(\cdot)$:} Accordingly, we denote $\mu(\cdot)$ as a mapping function from any vertex subset, $V'(q)$, of $V(q)$ to its mapping subset, $V'(g)$, of $V(g)$ (via the vertex-to-vertex mapping $M$). That is, we have $\mu(V'(q)) = V'(g)$, where any vertex $q_j \in V'(q)$ is mapped to a vertex $M(q_j) = v_i\in V'(g)$.


\noindent {\bf Neighbor Difference Semantics, $ND(q_j, v_i)$:} Let $N(v_i)$ be a set of 1-hop neighbors of vertex $v_i \in V(g)$ in the subgraph $g$. Similarly, $N(q_j)$ is a set of $q_j$'s 1-hop neighbors in the query graph $q$.

Then, for each vertex pair $(v_i, q_j)$ between $g$ and $q$, their \textit{neighbor difference}, $ND(q_j, v_i)$, is defined as the number of (matching) 1-hop neighbors (or edges) that $v_i$ is missing, based on the target vertex $q_j$ (and its neighbors). Formally, we have the following definition of the neighbor difference semantics.


\begin{definition} \textbf{(Neighbor Difference, \bm{$ND(q_j, v_i)$})} Given a target vertex, $q_j$, of a query graph $q$, a vertex, $v_i$, of a subgraph $g$, and a mapping function $\mu(\cdot)$ from any subset of $V(q)$ to its corresponding subset of $V(g)$ (via vertex-to-vertex mapping $M$), their \textit{neighbor difference}, $ND(q_j, v_i)$, is given by:
    \begin{equation}
        ND(q_j, v_i)=|\mu(N(q_j)) - N(v_i)|,\label{eq:ND}
    \end{equation}
where $N(\cdot)$ is a set of 1-hop neighbor vertices,  ``$-$'' is a \textit{set difference} operator, and $|\cdot|$ is the cardinality of a set. 
\label{def:neighbor_diff}
\end{definition}

Intuitively, in Definition \ref{def:neighbor_diff}, for each vertex pair $(q_j,v_i)$ in $q$ and $g$, the neighbor difference, $ND(q_j, v_i)$, is given by the number of missing edges $e(v_i, v_j)$ with an ending vertex $v_i$ (when their corresponding edges $e(q_j, q_j)$ in the target query graph $q$ exist).



\begin{example}
    (Continue with Example~\ref{example1}). In the previous example of Figure~\ref{fig:AND_application}, compared with the query vertex $q_1$ in query graph $q$, the data vertex $v_1$ in subgraph $g$ has one missing edge between $v_1$ and $v_3$ (while edge $e(q_1, q_3)$ exists in $q$). Thus, we have the neighbor difference $ND(q_1, v_1) = 1$. Similarly, since vertex $v_2$ has 1-hop neighbors $v_1$ and $v_4$ (while edges $e(q_2, q_1)$ and $e(q_2, q_4)$ exist in $q$), we have $ND(q_2, v_2) = 0$. \qquad $\blacksquare$


\end{example}




\noindent {\bf The Aggregation Over Neighbor Differences, $AND(q,g)$:} Next, we consider the \textit{aggregated neighbor differences} (AND), $AND(q, g)$, for vertex pairs $(q_j, v_i)$ from query graph $q$ and subgraph $g$, respectively. Intuitively, $AND(q, g)$ outputs an aggregation over the numbers of missing edges $e(v_i, v_j)$ (or 1-hop neighbors $v_j$ in $N(v_i)$) for all vertices $v_i$ in $g$, according to the targeted query graph $q$ (i.e., edges $e(q_j, q_j)$ in $q$).


\begin{definition}
    \textbf{(Aggregated Neighbor Difference, \bm{$AND(q,g)$})} Given a query graph $q$, a subgraph $g$, and a 1-to-1 vertex mapping $M$ from $V(q)$ to $V(g)$ (note: $|V(q)|=|V(g)|$), the \textit{aggregated neighbor difference}, $AND(q,g)$, between $q$ and $g$ is defined as the aggregation over neighbor differences of all the matching vertex pairs $(q_j, v_i)$, i.e.,
    \begin{equation}
        AND(q,g)=f\left(\left\{ ND(q_j, v_i) | \forall (q_j, v_i), s.t. M(q_j) = v_i \right\}\right),
        \label{eq:AND}
    \end{equation}
    where $ND(q_j, v_i)$ is given by Eq.~(\ref{eq:ND}), and $f(S)$ is an aggregate function (e.g., MAX, AVG, or SUM) over a set $S$.
\label{def:AND}
\end{definition}


In Definition \ref{def:AND}, the aggregated neighbor difference, $AND(q,g)$, is given by the aggregation over neighbor differences $ND(q_j, v_i)$ of all the matching pairs $(q_j, v_i)$ between $q$ and $v_i$. The aggregation function $f(S)$ can have different semantics such as MAX, AVG, or SUM. In Example \ref{example1} (i.e., the skilled team search), the MAX aggregate function returns the \textit{maximum} possible collaboration effort (i.e., $ND(q_j, v_i)$) that team members $v_i$ need (due to no collaboration experience with other team members before). Similarly, AVG (or SUM) aggregate function obtains the extra collaboration cost each team member has to spend on average (or the total collaboration cost for the entire team).




\begin{figure}\vspace{-2ex}
    \centering
    \subfigure[query graph $q$ and its matching subgraph $g$] {
        \includegraphics[height=3.8cm]{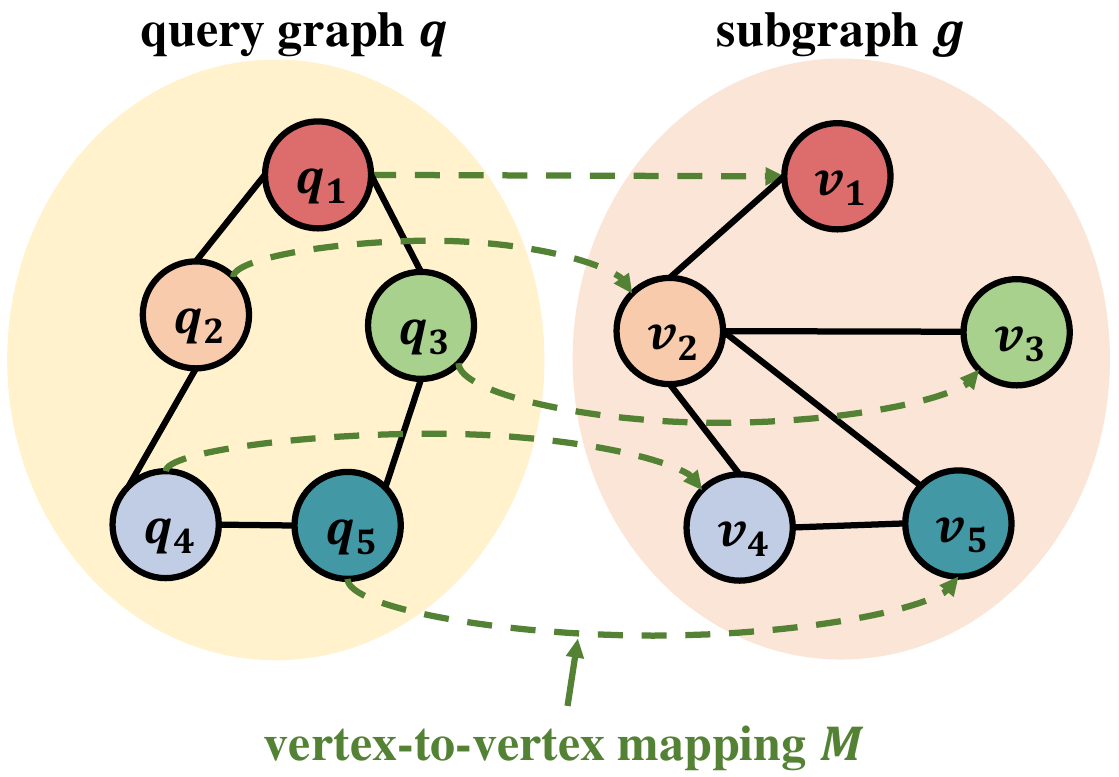}
        \label{subfig:fig2_a}
    }\\
    \subfigure[aggregated neighbor difference between $q$ and $g$]{
        \includegraphics[height=3cm]{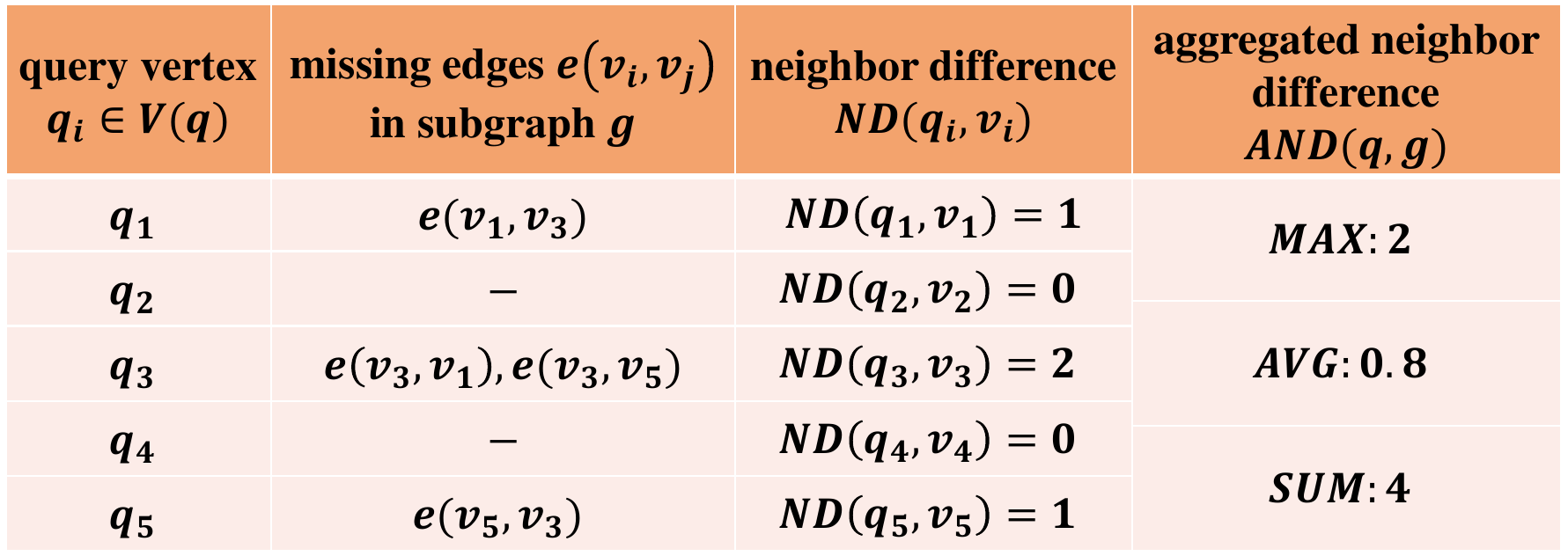}
        \label{subfig:fig_b}
    }
    \vspace{-3ex}
    \caption{An Example of the Aggregated Neighbor Difference.}
    \label{fig:AND}
\end{figure}

\begin{example}
Figure \ref{subfig:fig2_a} illustrates the vertex-to-vertex mapping $M$ from vertices of query graph $q$ to that of subgraph $g$ (as given in Example \ref{example1}), whereas Figure~\ref{subfig:fig_b} shows the neighbor differences, $ND(q_j, v_i)$, in subgraph $g$, for each query vertex $q_j$, and their aggregated neighbor differences $AND(q, g)$ under different semantics. 

Specifically, in Figure~\ref{subfig:fig_b}, we can see that vertex $v_1$ is not connected to $v_3$ in $g$, compared with the edge $e(q_1, q_3)$ in the target query graph $q$. Thus, we have $ND(q_1, v_1) = 1$. Similarly, we can compute neighbor differences for other vertex pairs $(q_j, v_i)$ (for $i\geq 2$). By aggregating these neighbor differences $ND(q_j, v_i)$ (for $1\leq i\leq 5$) with different aggregation functions $f(\cdot)$ (as given in Definition \ref{def:AND}), we can obtain their aggregated neighbor difference, $AND(q,g)$, that is, 2 ($=\max\{1, 0, 2, 0, 1\})$ for MAX, 0.8 $\big(= \frac{1+0+2+0+1}{5}\big)$ for AVG, and 4 ($=1+0+2+0+1$) for SUM. \qquad $\blacksquare$

\end{example}

Note that, since the AND score $AND(q, g)$ with the AVG aggregate is given by the AND score with the SUM aggregate divided by a constant (i.e., $|V(q)|$). In subsequent discussions, we will only focus on MAX and SUM aggregates for $f$ (note: AVG has the same S$^3$AND query answers as SUM).

\begin{table}[t!]
\caption{Symbols and Descriptions}
\vspace{-2ex}
\small
\label{tab1}
\begin{center}
\begin{tabular}{|l|p{6cm}|}
\hline
\textbf{Symbol}&{\textbf{Description}} \\
\hline\hline
$G$ & a data graph\\
\hline
$V(G)$ & a set of vertices $v_i$\\
\hline
$E(G)$ & a set of edges $e(u,v)$\\
\hline
$q$ & a query graph\\
\hline
$g$ & a subgraph of data graph $G$\\
\hline
$v_i.W$ & a keyword set of vertex $v_i$\\
\hline
$N(v_i)$ & a set of vertex $v_i$'s 1-hop neighbors\\
\hline
$ND(q_j, v_i)$ & the 1-hop neighbor difference between vertices $q_j$ and $v_i$\\
\hline
$AND(q, g)$ & the aggregation over all the neighbor differences of vertex pairs $(q_j, v_i)$ in graphs $q$ and $g$ \\
\hline
$\sigma$ & a threshold for the aggregated neighbor difference\\
\hline
\end{tabular}
\end{center}\vspace{2ex}
\end{table}

\vspace{-1ex}
\subsection{The S$^3$AND Problem Definition}
\label{sec:S3AND}
In this subsection, we formulate the \textit{subgraph similarity search problem under the aggregated neighbor difference semantics} (S$^3$AND).

\begin{definition}
    \label{def:S3AND}
    \textbf{(Subgraph Similarity Search Under Aggregated Neighbor Difference Semantics, \bm{$S^3AND(G,q)$})} Given a data graph $G$, a query graph $q$, a vertex-to-vertex mapping $M: V(q)\to V(g)$, and an aggregation threshold $\sigma$, a \textit{subgraph similarity search under the aggregated neighbor difference semantics} (S$^3$AND) retrieves connected subgraphs $g$ of $G$, such that:
    \begin{itemize}
        \item {\bf (Equal Subgraph Size)} $|V(q)|=|V(g)|$;
        \item {\bf (Keyword Set Containment)} for the mapping vertices $q_j\in V(q)$ and $v_i\in V(g)$, it holds that $q_j.W\subseteq v_i.W$, and;
        \item {\bf (Aggregated Neighbor Difference)} the aggregated neighbor difference satisfies the condition that $AND(q,g) \leq \sigma$,
    \end{itemize}
where $AND(q,g)$ is given by Eq.~(\ref{eq:AND}).
\end{definition}


Intuitively, in Definition \ref{def:S3AND}, the S$^3$AND problem retrieves all the subgraphs $g$ that satisfy the AND constraints, with respect to the query graph $q$. In particular, there exists a 1-to-1 vertex mapping, $M$, from each subgraph $g$ to query graph $q$. Thus, they have equal graph size, that is, $|V(q)|=|V(g)|$. Moreover, for the mapping vertices $q_j$ and $v_i$ from graphs $q$ and $g$, respectively, their associated keyword sets satisfy the containment relationship, that is, $q_j.W\subseteq v_i.W$. Further, their aggregated neighbor difference $AND(q,g)$ should be as low as possible (i.e., $AND(q,g) \leq \sigma$). The three conditions above guarantee that the subgraphs $g$ can maximally match with the required target graph pattern $q$.

In Definition~\ref{def:S3AND}, we used the constraint of the \textit{Keyword Set Containment}. In practice, we may also consider other keyword matching constraints such as keyword embedding similarity, ontology similarity, and so on, and adapt our proposed techniques (e.g., pruning and indexing) to handle such keyword matching constraints. 
Moreover, the AND score, $AND(q, g)$, considers the (mis)matching of edges between query/data vertices and their 1-hop neighbors. As in Example \ref{example1}, the AND score implies the communication overhead between team members and their collaborators (i.e., 1-hop neighbors in collaboration networks).
We would like to leave interesting topics of considering variants of S$^3$AND query semantics (e.g., with different keyword matching or topological similarity options) as our future work.




\noindent {\bf Challenges:} A straightforward method to answer the S$^3$AND query is to enumerate all possible subgraphs $g$ in the data graph $G$, compute the aggregated neighbor difference $AND(q, g)$ between each subgraph $g$ and query graph $q$, and return all S$^3$AND query answers with $AND(q, g)$ lower than threshold $\sigma$. However, this straightforward method is rather inefficient, due to a large number of candidate subgraphs within large-scale data graph $G$ and high refinement costs (w.r.t. vertex mapping, keywords, and AND computations). Therefore, it is quite challenging to process S$^3$AND queries efficiently and effectively.



Table \ref{tab1} depicts the commonly used notations and their descriptions in this paper.

\section{The S$^3$AND Processing Framework}
\label{sec:framework}

Algorithm~\ref{alg:framework} illustrates  a general framework for S$^3$AND query answering in a large-scale data graph $G$. Figure~\ref{fig:framework} provides a visual workflow of the pseudo code in Algorithm~\ref{alg:framework}, which consists of two phases, \textit{offline pre-computation} (lines 1-3 of Algorithm~\ref{alg:framework}) and \textit{online S$^3$AND query processing phases} (lines 4-7 of Algorithm~\ref{alg:framework}).

Specifically, as illustrated in Figure~\ref{fig:framework}, in the \textit{offline pre-computation phase}, we offline pre-compute some auxiliary data, $v_i.Aux$, of each vertex $v_i$ in large-scale data graph $G$ (lines 1-2 of Algorithm~\ref{alg:framework}), and construct a tree index $\mathcal{I}$ over these pre-computed data $v_i.Aux$  to facilitate online query optimizations like pruning (line 3 of Algorithm~\ref{alg:framework}). 

In the \textit{online S$^3$AND query processing phase}, for each S$^3$AND query, we traverse the tree index $\mathcal{I}$ by applying our proposed pruning strategies (e.g., the keyword set and AND lower bound pruning) to retrieve candidate vertices w.r.t. query vertices $q_j$ in the query graph $q$ (lines 4-5 of Algorithm~\ref{alg:framework}).
Next, we assemble candidate vertices of query vertices $q_j$ and obtain candidate subgraphs $g$ (line 6 of Algorithm~\ref{alg:framework}). 
Finally, we refine candidate subgraphs $g$ and return a set, $S$, of actual S$^3$AND subgraph answers (line 7 of Algorithm~\ref{alg:framework}).

\section{Pruning Strategies}
\label{sec:pruning}

In this section, we present effective pruning strategies that reduce the problem search space during the online S$^3$AND query processing phase (lines 5-7 of Algorithm~\ref{alg:framework}).

\subsection{Keyword Set Pruning}
\label{subsec:keyword_pruning}


In Definition~\ref{def:S3AND}, the keyword set $v_i.W$ of each vertex $v_i$ in the S$^3$AND subgraph answer $g$ must be a superset of the keyword set $q_j.W$ for its corresponding query vertex $q_j$ in the query graph $q$. Based on this, we design an effective \textit{keyword set pruning} method to rule out candidate vertices that do not satisfy this keyword set constraint.

\begin{algorithm}[!t]
    \caption{\bf The $S^3AND$ Processing Framework}
    \label{alg:framework}
    \KwIn{
        \romannumeral1) a data graph $G$, 
        \romannumeral2) a query graph $q$, and
        \romannumeral3) an aggregated neighbor difference threshold $\sigma$
    }
    \KwOut{
        a set, $S$, of subgraphs $g$ matching with the query graph $q$ under AND semantics
    }
    
    \tcp{\bf offline pre-computation phase}

    \For{each $v_i \in V(G)$}{
    


        

        compute the auxiliary data $v_i.Aux$
        
    }

    \textcolor{black}{construct a tree index $\mathcal{I}$ over pre-computed aggregate data in graph $G$}

    \tcp{\bf online S$^3$AND query processing phase}

    \For{each S$^3$AND query}{

        
        \textcolor{black}{traverse the tree index $\mathcal{I}$ by applying the keyword set and AND lower bound pruning strategies to retrieve candidate vertices w.r.t. query vertices $q_j$ in the query graph $q$}

        \textcolor{black}{assemble candidate vertices of query vertices $q_j$ and obtain candidate subgraphs $g$}

        \textcolor{black}{refine candidate subgraphs $g$ and return a set, $S$, of actual $S^3AND$ subgraph answers}
        
        
    }
    
\end{algorithm}%
\begin{figure}[t]
    \centering\vspace{-2ex}
    \includegraphics[width=1.0\linewidth]{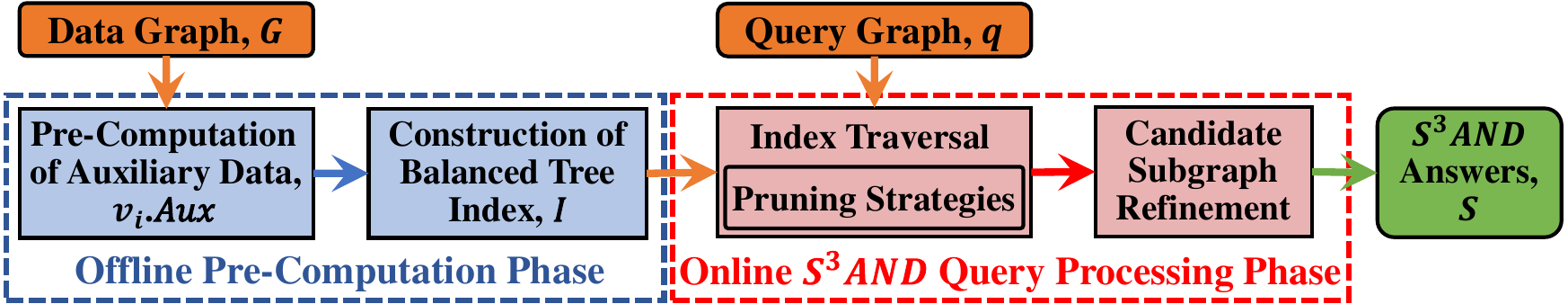}
    \vspace{-3ex}
    \caption{The workflow of S$^3$AND query processing.}
    \label{fig:framework}
\end{figure}


\begin{lemma}
    \label{lemma:keyword_pruning}
    {\bf (Keyword Set Pruning)} Given a candidate vertex $v_i$ and a query graph $q$, vertex $v_i$ can be safely pruned, if it holds that: $v_i.W \cap q_j.W \neq q_j.W$ (i.e., $q_j.W  \not\subseteq v_i.W$), for all $q_j \in V(q)$.
\end{lemma}
\begin{proof}
    For all query vertices $q_j \in V(q)$, if $v_i.W \cap q_j.W \neq q_j.W$ holds for a candidate vertex $v_i$, it indicates that query keyword sets $q_j.W$ are not subsets of $v_i.W$. Thus, according to the keyword set containment property in Definition \ref{def:S3AND}, we can infer that vertex $v_i$ cannot match with any query vertex $q_j$ in the query graph $q$. Hence, we can safely prune vertex $v_i$, which completes the proof. 
\end{proof}


\noindent{\bf Discussions on How to Implement the Keyword Set Pruning:} In order to enable the keyword set pruning (as given in Lemma~\ref{lemma:keyword_pruning}), we can offline pre-compute a bit vector, $v_i.BV$, of size $B$ for the keyword set $v_i.W$ in each vertex $v_i$. In particular, we first initialize the bit vector $v_i.BV$ with $\vv{0}$, and then hash each keyword in $v_i.W$ into a bit position in $v_i.BV$ (via a hashing function; setting the position to ``1''). The case of computing query bit vector $q_j.BV$ for query keyword set $q_j.W$ (w.r.t. query vertex $q_j$) is similar.

As a result, the pruning condition, $v_i.W \cap q_j.W \neq q_j.W$, in the keyword set pruning method can be written as: 
\begin{eqnarray}
    v_i.BV \, \land \, q_j.BV \neq q_j.BV,
    \label{eq:keyword_set_pruning_condition2}
\end{eqnarray}
where ``$\land$'' is a bit-AND operator between two bit vectors.

\noindent{\bf Enhancing the Pruning Power via Keyword Grouping:} 
Since the keyword domain of the real data may be large, the size, $B$, of bit vectors $v_i.BV$ is much smaller than the keyword domain size, which may lead to hashing conflicts (i.e., different keywords are hashed to the same bit position in $v_i.BV$). In order to enhance the pruning power of keyword set pruning, we propose a \textit{keyword grouping} optimization approach, which can reduce the probability of incurring false positives via keyword set pruning. 

Specifically, we divide the keyword domain into $m$ disjoint groups. For each vertex $v_i$, if a keyword in $v_i.W$ falls into the $x$-th keyword group, we will hash this keyword into the $x$-th bit vector  $v_i.BV^{(x)}$ (for $1 \leq x \leq m$) via a hashing function. This way, the pruning condition, $v_i.W \cap q_j.W \neq q_j.W$, in the keyword set pruning (Lemma \ref{lemma:keyword_pruning}) can be rewritten as:

\begin{equation}
    \bigvee_{x=1}^m \left(v_i.BV^{(x)} \: \bigwedge \: q_j.BV^{(x)} \neq q_j.BV^{(x)}\right), 
    \label{eq:keyword_set_pruning_condition3}
\end{equation}
where $v_i.BV^{(x)}$ and $q_j.BV^{(x)}$ are bit vectors with the hashed keywords from the $x$-th keyword group in $v_i.W$ and $q_j.W$, respectively.


\subsection{AND Lower Bound Pruning}
\label{subsec:AND_bound_pruning}

{\color{black}

According to Definition~\ref{def:S3AND}, the aggregated neighbor difference (AND) between a subgraph $g$ and a query graph $q$ must satisfy the AND constraint, that is, $AND(q,g) \leq \sigma$. Therefore, we present an \textit{AND lower bound pruning} method, which effectively filters out candidate subgraphs with high AND values below.


\begin{lemma}
    \label{lemma:AND_bound_pruning}
    {\bf (AND Lower Bound Pruning)} Given a candidate subgraph $g$, a query graph $q$, and an aggregated neighbor difference threshold $\sigma$, subgraph $g$ can be safely pruned, if it holds that $lb\_AND(q,g) > \sigma$, where $lb\_AND(q,g)$ is a lower bound of $AND(q,g)$.
\end{lemma}
\begin{proof}
Since $lb\_AND(q,g)$ is a lower bound of the aggregated neighbor difference $AND(q,g)$, we have $AND(q,g) \geq lb\_AND(q,g)$. From the lemma assumption that $lb\_AND(q,g) > \sigma$ holds, by the inequality transition, it holds that $AND(q,g) \geq lb\_AND(q,g) > \sigma$. Thus, from Definition~\ref{def:S3AND}, candidate subgraph $g$ cannot be the S$^3$AND answer and can be safely pruned, which completes the proof.
\end{proof}



}

{\color{black}
\noindent{\bf Discussions on How to Compute an AND Lower Bound,\\ $lb\_AND(q,g)$: } Based on Eq.~(\ref{eq:AND}), in order to compute a lower bound, $lb\_AND(q,g)$, of the AND score $AND(q,g)$, we only need to obtain a lower bound, $lb\_ND(q_j, v_i)$, of the neighbor difference $ND(q_j, v_i)$ (as given in Eq.~(\ref{eq:ND})) for each matching vertex pair $(q_j, v_i)$. This way, we have: 
\begin{equation}
    lb\_AND(q,g)=f\left(\left\{ lb\_ND(q_j, v_i) | \forall (q_j, v_i), s.t. M(q_j) = v_i \right\}\right).
    \label{eq:lb_AND}
\end{equation}

\underline{\it The Computation of the Neighbor Difference Lower Bound}\\ \underline{$lb\_ND(q_j, v_i)$.} To compute a lower bound $lb\_ND(q_j, v_i)$ of the neighbor difference $ND(q_j, v_i)$, we can rewrite the neighbor difference $ND(q_j, v_i)$ in Eq.~(\ref{eq:ND}) as follows: 
\begin{eqnarray}
ND(q_j, v_i)&=&|\mu(N(q_j)) - N(v_i)|\label{eq:ND2} \\
&=&|\mu(N(q_j))|-|\mu(N(q_j)) \cap N(v_i)|.  \notag
\end{eqnarray}

In Eq.~(\ref{eq:ND2}), the first term $|\mu(N(q_j))|$ is a constant during online S$^3$AND query processing (i.e., the number of vertex $q_j$'s neighbors in the query graph $q$). Thus, in order to calculate the ND lower bound $lb\_ND(q_j, v_i)$, we alternatively need to compute an upper bound of the second term in Eq.~(\ref{eq:ND2}) (i.e., $|\mu(N(q_j)) \cap N(v_i)|$). Since $N(v_i) \supseteq (\mu(N(q_j)) \cap N(v_i))$ holds, we can obtain its upper bound: $|N(v_i)| \geq |\mu(N(q_j)) \cap N(v_i)|$.

In other words, we have the ND lower bound below: 
\begin{eqnarray}
lb\_ND(q_j, v_i) = \max\{0, |\mu(N(q_j))| - |N(v_i)|\}.
\label{eq:lb_ND1}
\end{eqnarray}


\underline{\it The Computation of a Tighter Neighbor Difference Lower Bound}\\ \underline{$lb\_ND(q_j, v_i)$.} Note that, some neighbors, $q_l$, of query vertex $q_j$ may not match with that, $v_l$, of data vertex $v_i$ with respect to their keyword sets (i.e., $q_l.W \nsubseteq v_l.W$).
Therefore, $|N(v_i)|$ may not be a tight upper bound of $|\mu(N(q_j)) \cap N(v_i)|$, and in turn $lb\_ND(q_j, v_i)$ in Eq.~(\ref{eq:lb_ND1}) is not a tight neighbor difference lower bound.

Below, we will consider the keyword set matching between (neighbors of) vertices $v_i$ and $q_j$, and derive a tighter neighbor difference lower bound. 
Specifically, for each neighbor $q_l$ of query vertex $q_j$, if its keyword set $q_l.W$ is a subset of the union of keyword sets from $N(v_i)$ (i.e., $v_i$'s neighbors), we can count 1, for the upper bound of $|\mu(N(q_j)) \cap N(v_i)|$). Formally, we have this upper bound given by: 
\begin{eqnarray}
\sum_{q_l\in N(q_j)} \Phi\left(q_l.W \subseteq \cup_{\forall v_l \in N(v_i)} v_l.W\right) \geq |\mu(N(q_j)) \cap N(v_i)|,\notag
\end{eqnarray}
where $\Phi(\cdot)$ is an indicator function (i.e., $\Phi(z) = 1$, if $z$ is \textit{true}; $\Phi(z) = 0$, otherwise).

As a result, we can obtain a tighter ND lower bound below:
\begin{equation}
    lb\_ND(q_j,v_i) = |\mu(N(q_j))|- \hspace{-3ex}\sum_{q_l\in N(q_j)} \hspace{-3ex}\Phi\left(q_l.W \subseteq \cup_{\forall v_l \in N(v_i)} v_l.W\right).\label{eq:lb_ND2}
\end{equation}

To efficiently check the containment of two keyword sets in Eq.~(\ref{eq:lb_ND2}) (i.e., $q_l.W \subseteq \cup_{\forall v_l \in N(v_i)} v_l.W$), we can also use their keyword bit vectors to replace the parameter of the indicator function $\Phi(\cdot)$ in  Eq.~(\ref{eq:lb_ND2}), that is,
\begin{eqnarray}
    &&\hspace{-5ex}lb\_ND(q_j, v_i) = |\mu(N(q_j))| \label{eq:lb_ND3}\\
    &&- \hspace{-2ex}\sum_{q_l\in N(q_j)} \hspace{-1ex}\Phi\left(\bigwedge_{x=1}^{m} \left( q_l.BV^{(x)} \bigwedge \left(\bigvee_{\forall v_l \in N(v_i)} v_l.BV^{(x)}\right) = q_l.BV^{(x)}\right)\right),\hspace{-3ex}\notag
\end{eqnarray}
where $\land$ and $\lor$ are bit-AND and bit-OR operators between two bit vectors, respectively.

\noindent {\bf ND Lower Bound Pruning for Individual Vertices:} Note that, Lemma~\ref{lemma:AND_bound_pruning} uses AND lower bound, $lb\_AND(q, g)$, to prune the entire candidate subgraphs $g$ (which are however not available during the filtering phase). Therefore, in the sequel, we will discuss how to utilize ND lower bounds, $lb\_ND(q_j, v_i)$ (w.r.t. individual vertices $v_i$), to filter out false alarms of vertices $v_i$ (or retrieve candidate vertices $v_i$), for different aggregation functions $f(\cdot)$ (e.g., $MAX$ or $SUM$).


\underline{\it ND Lower Bound Pruning on Individual Vertices.} From Lemma~\ref{lemma:AND_bound_pruning} and Eq.~(\ref{eq:lb_AND}), we can derive that: a candidate vertex $v_i$ can be safely pruned (for either $MAX$ or $SUM$), if its ND lower bound $lb\_ND(q_j,v_i)$ is greater than the aggregate threshold $\sigma$. Formally, we have the corollary below.

\begin{corollary} {\bf (ND Lower Bound Pruning)}
Given a query vertex $q_j \in V(q)$ and an aggregate threshold $\sigma$, a vertex $v_i$ can be safely pruned, if it holds that  $lb\_ND(q_j, v_i) > \sigma$.
\label{coro:ND_pruning_MAX_SUM}
\end{corollary}

}

\section{Offline Pre-Computation}
In this section, we discuss how to offline pre-compute data over a data graph $G$ to enable effective pruning as discussed in Section~\ref{sec:pruning}, and construct an index $\mathcal{I}$ over the pre-computed data. 

\subsection{Offline Pre-Computed Auxiliary Data}


To facilitate efficient online S$^3$AND computation, Algorithm~\ref{alg:offline} offline pre-computes relevant aggregation information for each vertex in graph $G$, which can be used for pruning candidate vertices/subgraphs during the S$^3$AND query processing.


\begin{algorithm}[!t]
    \caption{\bf Offline Pre-Computation of Auxiliary Data}
    \label{alg:offline}
    \KwIn{
        \romannumeral1) a data graph $G$, and
        \romannumeral2) the number, $m$, of keyword groups
    }
    \KwOut{
        the pre-computed auxiliary data $v_i.Aux$ for each vertex $v_i$
    }

    \For{each vertex $v_i \in V(G)$}{
        

        \tcp{keyword bit vectors}
        \For{keyword group $x = 1$ to $m$}{

            hash the keywords in the $x$-th keyword group of $v_i.W$ into a bit vector $v_i.BV^{(x)}$ of size $B$
        

        

            

        }    

        

        
    }

    \For{each vertex $v_i \in V(G)$}{

        \tcp{neighbor keyword bit vectors}

        initialize neighbor keyword bit vectors $v_i.NBV^{(x)}$  with $\vv{0}$ (for $1\leq x\leq m$)
        
        \For{each neighbor vertex $v_l \in N(v_i)$}{        
            \For{keyword group $x = 1$ to $m$}{
                $v_i.NBV^{(x)} = v_i.NBV^{(x)} \: \lor \: v_l.BV^{(x)}$
            
            }
            }

        \tcp{the number of distinct neighbor keywords}
        
        $v_i.nk=|\cup_{\forall v_l\in N(v_i)} v_l.W|$

        \tcp{obtain the auxiliary data structure $v_i.Aux$}
        $v_i.Aux = \big(v_i.BV^{(x)}, v_i.NBV^{(x)}, v_i.nk\big)$
    }
    \Return $v_i.Aux$
\end{algorithm}


Specifically, for each vertex $v_i \in V(G)$, we maintain an auxiliary data structure $v_i.Aux$ (lines 1-10). Since the domain size of keywords can be quite large, to improve the pruning power, we divide the keyword domain into $m$ disjoint groups to reduce the chance of false positives via keyword bit vectors. Then, for each vertex $v_i$,  we hash the keywords in $v_i.W$ that fall into the $x$-th keyword group (for $1 \leq x\leq m$), and obtain a keyword bit vertor $v_i.BV^{(x)}$ with size $B$ (lines 2-3). Next, for each vertex $v_i$, we first initialize neighbor keyword bit vectors $v_i.NBV^{(x)} \: (1 \leq x\leq m)$ with $\vv{0}$, and then conduct the bit-OR operation over bit vectors, $v_l.BV^{(x)}$, of all $v_i$'s neighbor vertices $v_l \in N(v_i)$ to compute neighbor keyword bit vectors $v_i.NBV^{(x)}$ (lines 4-8). We also count the number, $v_i.nk$, of distinct keywords from $v_i$'s neighbors in $N(v_i)$ (line 9). This way, we store $v_i.BV^{(x)}$, $v_i.NBV^{(x)}$, and $v_i.nk$ in the auxiliary data structure $v_i.Aux$ (line 10). Finally, we return all the pre-computed auxiliary data $v_i.Aux$ (line 11).


To summarize, $v_i.Aux$ contains the following information:
\begin{itemize}[noitemsep, topsep=0pt, partopsep=0pt, itemsep=1pt, parsep=1pt]
    \item {\bf $m$ keyword bit vectors, \boldsymbol{$v_i.BV^{(x)}$} (for \boldsymbol{$1 \leq x \leq m$}), of size \boldsymbol{$B$},} which is obtained by using a hashing function $f(w)$ to hash each keyword $w \in v_i.W$ of each group to an integer between $[0, B-1]$ and set the $f(w)$-th bit position to 1 (i.e., $v_i.BV^{(x)}[f(w)]=1$);
    \item {\bf $m$ neighbor keyword bit vectors, \boldsymbol{$v_i.NBV^{(x)}$} (for \boldsymbol{$1 \leq x \leq m$}),} which is computed by aggregating all keywords in $v_l.W$ from neighbor vertices $v_l \in N(v_i)$ (i.e., $v_i.NBV^{(x)}=\lor_{\forall v_l\in N(v_i)} v_l.BV^{(x)}$), and;
    \item {\bf the number, \boldsymbol{$v_i.nk$}, of distinct neighbor keywords,} which is given by counting the number of distinct keywords from neighbors $v_l\in N(v_i)$ of vertex $v_i$ (i.e., $v_i.nk=|\cup_{\forall v_l\in N(v_i)} v_l.W|$).
\end{itemize}


\vspace{-1ex}
\subsection{Indexing Mechanism}
\label{sec:index}

In this subsection, we show the offline construction of a tree index $\mathcal{I}$ on data graph $G$ with the pre-computed auxiliary data in detail to support online S$^3$AND query computation. 

\noindent
\textbf{The Data Structure of Index $\mathcal{I}$:}
We will construct a tree index $\mathcal{I}$ on the data graph $G$.
Specifically, the tree index $\mathcal{I}$ contains two types of nodes, leaf and non-leaf nodes.

\underline{\textit{Leaf Nodes:}} 
Each leaf node $\mathcal{N}$ contains a vertex set in the data graph $G$. Each vertex $v_i$ is associated with the following pre-computed data in $v_i.Aux$:

\begin{itemize}[noitemsep, topsep=1pt, partopsep=1pt, itemsep=1pt, parsep=1pt]
    \item $m$ keyword bit vectors, $v_i.BV^{(x)}$ (for $1 \leq x \leq m$);
    \item $m$ neighbor keyword bit vectors, $v_i.NBV^{(x)}$ (for $1 \leq x \leq m$), and;
    \item the number of distinct neighbor keywords, $v_i.nk$.
\end{itemize}

\underline{\textit{Non-Leaf Nodes:}}
Each non-leaf node $\mathcal{N}$ has multiple entries $\mathcal{N}_i$, each corresponding to a subgraph of $G$.
Each entry $\mathcal{N}_i$ is associated with the following aggregates:

\begin{itemize}[noitemsep, topsep=1pt, partopsep=1pt, itemsep=1pt, parsep=1pt]
    \item a pointer to a child node, 
    $\mathcal{N}_i.ptr$;
    \item $m$ aggregated keyword bit vectors, $\mathcal{N}_i.BV^{(x)}$\\ $=$ $\lor_{\forall v_l \in \mathcal{N}_i} v_l.BV^{(x)}$ (for $1 \leq x \leq m$);
    \item $m$ aggregated neighbor keyword bit vectors, 
    $\mathcal{N}_i.NBV^{(x)} = \lor_{\forall v_l \in \mathcal{N}_i} v_l.NBV^{(x)}$  (for $1 \leq x \leq m$), and;
    \item the maximum number of distinct neighbor keywords for vertices $v_l$ under entry $\mathcal{N}_i$, that is,
    $\mathcal{N}_i.nk = \max_{\forall v_l \in \mathcal{N}_i} v_l.nk$.
\end{itemize}

\noindent \textbf{Index Construction:} Algorithm~\ref{alg:index2} illustrates the pseudo code of constructing a balanced tree index $\mathcal{I}$ in a top-down manner.

Specifically, given the fanout, $fanout$, of index nodes, we first calculate the tree height $h = \lceil log_{fanout}(|V(G)|)\rceil$, and use the tree root $root(\mathcal{I})$ to represent the entire vertex set $V(G)$ (lines 1-2).
Then, we construct the index $\mathcal{I}$ in a top-down fashion (i.e., level $l$ from $h$ to $1$). 
In particular, for each node $\mathcal{N}^{(l)}$ on the $l$-th level of index $\mathcal{I}$, we invoke a cost-model-based partitioning function, {\sf CM\_Partitioning}$(\mathcal{N}^{(l)}, fanout)$, to obtain $fanout$ partitions $\mathcal{N}_i^{(l-1)}$ ($1\leq i \leq fanout$) of similar sizes as child nodes (lines 3-5).
After we partition each non-leaf node on level $l=1$, we obtain leaf nodes on level 0 and complete the construction of index $\mathcal{I}$. Finally, we return this balanced index $\mathcal{I}$ (line 6).


\begin{algorithm}[!t]
    \caption{\bf The Balanced Tree Index Construction}
    \label{alg:index2}
    \KwIn{
        \romannumeral1) the pre-computed auxiliary data $v_i.Aux$ over a data graph $G$, and
        \romannumeral2) the fanout, $fanout$, of the index node
    }
    \KwOut{a balanced tree index, $\mathcal{I}$, of data graph $G$}

    tree height $h = \lceil log_{fanout}(|V(G)|)\rceil$

    tree root $root(\mathcal{I}) = V(G)$

    \For{level $l = h$ to 1}{
        \For{each node $\mathcal{N}^{(l)}$ on the $l$-th level of index $\mathcal{I}$}{
            \tcp{cost-model-based top-down partitioning}
            invoke function {\bf CM\_Partitioning}$(\mathcal{N}^{(l)}, fanout)$ to obtain $fanout$ partitions $\mathcal{N}_i^{(l-1)}$ ($1\leq i \leq fanout$) of similar sizes as child nodes 
        }
    }
    

    
    
    \Return $\mathcal{I}$
    
\end{algorithm}

\noindent
\textbf{A Construction Example of Tree Index $\mathcal{I}$:} We use Example~\ref{ex:index} to clearly illustrate the index construction process.

\begin{figure}[t]
    \centering
    \includegraphics[width=0.7\linewidth]{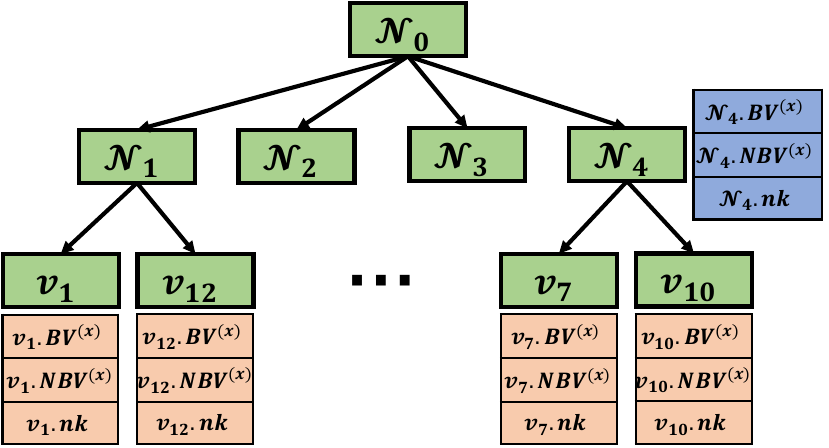}
    \caption{An example of constructing a tree index $\mathcal{I}$ based on Figure~\ref{fig:AND_application} ($n$=4, $\gamma$=0.2).}
    \label{fig:indexShow}
\end{figure}

\begin{example}
\label{ex:index}
    \textbf{(The Construction of a Tree Index, $ \mathcal{I}$)} Figure~\ref{fig:indexShow} illustrates a tree index $\mathcal{I}$ over a data graph $G$ in the example of Figure~\ref{fig:AND_application}, via cost-model-based graph partitioning (i.e., Algorithm \ref{alg:cm}), where $n=4$ and $\gamma=0.2$. 
    In particular, the tree root $\mathcal{N}_0$ (i.e., graph $G$) contains 4 leaf entries (i.e., 4 subgraph partitions, $\mathcal{N}_{1}$, $\mathcal{N}_{2}$, $\mathcal{N}_{3}$, and $\mathcal{N}_{4}$, respectively). Each entry (e.g., $\mathcal{N}_1$) is associated with a branch pointer $\mathcal{N}_1.ptr$ and some aggregates (i.e., $\mathcal{N}_1.BV^{(x)}$, $\mathcal{N}_1.NBV^{(x)}$, and $\mathcal{N}_1.nk$).
    
    Similarly, each leaf node (e.g., $\mathcal{N}_1$) contains a set of vertices (e.g., $v_1$ and $v_{12}$), where each vertex, say $v_1$, is associated with auxiliary data, for example, $v_1.Aux = (v_1.BV^{(x)}, v_1.NBV^{(x)}, v_1.nk)$. 
\end{example}


\noindent
\textbf{The Cost Model for Vertex Partitioning:}  In line 5 of Algorithm \ref{alg:index2}, we need to divide a set of vertices into $n$ partitions of similar sizes, by invoking function {\sf CM\_Partitioning $(\cdot, \cdot)$}. Since different partitioning strategies may result in different pruning effects, our goal is to propose a formal cost model to guide such partitioning. Intuitively, we would like to group those vertices with similar keyword bit vectors in the same partitions (achieving high pruning power), and dissimilar bit vectors in different partitions. We use a cost model, $Cost(Par)$, to evaluate the ``goodness'' (quality) of the partitioning strategy $Par$ as follows.
\setlength{\abovedisplayskip}{0pt}
\setlength{\belowdisplayskip}{0pt}
\begin{equation}
    Cost(Par) = \frac{\overbrace{\sum_{i=1}^{n} \sum_{v \in Par_{i}} \sum_{x=1}^{m} dist(v.BV^{(x)},c_i.BV^{(x)})}^{intra-partition \: distance \: \downarrow}}{\underbrace{\sum_{1\leq a< b\leq n} \sum_{x=1}^{m} dist(c_a.BV^{(x)},c_b.BV^{(x)}) +1}_{inter-partition \: distance \: \uparrow}},
    \label{eq:score}
\end{equation}

where $c_i$ is the center of partition $Par_i$ (i.e., mean of all bit vectors in this partition), and $dist(v.BV^{(x)},c_i.BV^{(x)})$ is given by the \textit{$L_1$-norm distance} \cite{malkauthekar2013analysis} between vectors $v.BV^{(x)}$ and $c_i.BV^{(x)}$ (note: in the special case of bit vectors, this is the \textit{Hamming distance} \cite{liu2011large}).

Here, we have:
\begin{eqnarray}
    dist(v.BV^{(x)},c_i.BV^{(x)}) =\sum_{e=1}^{B} \left|v.BV^{(x)}[e] - c_i.BV^{(x)}[e]\right|,
\end{eqnarray}
where $v.BV^{(x)}[e]$ is the $e$-th position in the vector $v.BV^{(x)}$, and $\Phi(\cdot)$ is an indicator function (i.e., $\Phi(z) = 1$, if $z$ is $true$; $\Phi(z) = 0$, otherwise).

Intuitively, the lower value of the cost $Cost(Par)$ indicates the good quality of the partitioning strategy (i.e., with small intra-partition distances and large inter-partition distances). Thus, we aim to find a good vertex partitioning approach that minimizes the cost function $Cost(Par)$.

\begin{algorithm}[h]
    \caption{ {\bf CM\_Partitioning}}
    \label{alg:cm}

    \KwIn{
        \romannumeral1) a set, $Par$, of vertices (or an index node) to be partitioned, 
        \romannumeral2) the number, $n$, of partitions (or the fanout of the index node),
        \romannumeral3) the number, $global\_iter$, of global iterations, and
        \romannumeral4) the number, $local\_iter$, of local iterations

    }
    \KwOut{a set, $global\_Par$, of $n$ partitions}

    $global\_cost = +\infty;$

    \For{$k = 1 \: to \: global\_iter$}{

        randomly select $n$ center vertices $C = \{c_1, c_2, \cdots, c_n\}$, and assign each vertex $v$ in $Par$ to a partition, $Par_i$, with the closest distance $\sum_{x=1}^{m} dist(v.BV^{(x)}, c_i.BV^{(x)})$
        
        
                obtain an initial partitioning strategy $local\_Par$ $=$ $\{Par_1, Par_2, \cdots, Par_n\}$ with cost $local\_cost= Cost(local\_Par)$ \quad// {\it Eq.~(\ref{eq:score})}


    \For{$j = 1 \: to \: local\_iter $}{

        \tcp{update $n$ center vertices}

        \For{$i = 1 \: to \: n$}{

            $c_i.BV^{(x)} = (\sum_{v \in Par_i} v.BV^{(x)})/|Par_i|$ (for $1 \leq x \leq m$)
        
        }



        \For{each vertex $v \in Par$}{

            assign $v$ to a partition $Par'_i$ ($|Par'_i| \leq (1+\gamma) \cdot |Par| / n  $) with the closest distance $\sum_{x=1}^{m} dist(v.BV^{(x)}, c_i.BV^{(x)})$
            



                    

                
                


        }
        


            
                
            

        

        


        

            obtain a new partitioning strategy $local\_Par'=\{Par'_1, Par'_2, \cdots, Par'_n\}$ with new cost $local\_cost' = Cost(local\_Par')$ \quad// {\it Eq.~(\ref{eq:score})}
    
            \If{$local\_cost' < local\_cost$}{
                
                $local\_Par$$\leftarrow$$local\_Par',$ $local\_cost$$\leftarrow$$local\_cost'$
                
            }

            
        
        }

    \If{$local\_cost < global\_cost$}{

        $global\_Par \leftarrow local\_Par$, $global\_cost \leftarrow local\_cost$
        
    }
    
}

    \Return $global\_Par$
\end{algorithm}

\noindent \textbf{Cost-Model-Based Partitioning Function, {\sf CM\_Partitioning $(\cdot, \cdot)$}:} 
Algorithm~\ref{alg:cm} illustrates 
the partitioning process based on our proposed cost model. Specifically, to avoid local optimality, we run our partitioning algorithm for $global\_iter$ iterations to achieve the lowest cost $global\_cost$ (lines 1-14). For each iteration, we start with a set of $n$ random initial center vertices $C = \{c_1, c_2, \cdots, c_n\}$ and assign each vertex $v$ in $Par$ to a partition $Par_i$ with the closest center vertex $c_i$ (for $1\leq i\leq n$; line 3). This way, we can obtain an initial partitioning strategy $local\_Par$ with the cost $local\_cost = Cost(local\_Par)$ (as given in Eq.~(\ref{eq:score}); line 4). 

Then, we will iteratively update center vertices $c_i$ and in turn their corresponding partitions $Par_i$, by minimizing intra-partition distances and maximizing inter-partition distances (in light of our cost model in Eq.~(\ref{eq:score}); lines 5-12). In particular, for $local\_iter$ iterations, we update the bit vectors, $c_i.BV^{(x)}$, of center vertices $c_i$ (for $1 \leq x \leq m$), by taking the mean of bit vectors for all vertices in each partition $Par_i$ (lines 6-7). We then assign vertices $v \in Par$ to a partition $Par_i'$ with the distance $\sum_{x=1}^{m} dist(v.BV^{(x)}, c_i.BV^{(x)})$ closest to the updated centers $c_i$, satisfying the constraint of the balanced partitions (i.e., $|Par'_i| \leq (1+\gamma) \cdot |Par| / n $), where $\gamma$ is a relaxation coefficient for the partition size (lines 8-9). As a result, we obtain a new partitioning strategy, $local\_Par'$, with cost $local\_cost' = Cost(local\_Par')$ (line 10). If this new partitioning $local\_Par'$ has lower cost (i.e., $local\_cost' < local\_cost$), we will accept this new partitioning strategy by letting $local\_Par = local\_Par'$ with the lower cost $local\_cost = local\_cost'$ (lines 11-12). After $local\_iter$ iterations, we will update $global\_Par$ and $global\_cost$ with the best partitioning strategy so far and its cost, respectively  (lines 13-14). Finally, we return the best partitioning strategy $global\_Par$ (after $global\_iter$ iterations) to obtain $n$ good-quality partitions (line 15).

\noindent
\textbf{Time Complexity Analysis:} 
For Algorithm~\ref{alg:offline}, for each vertex $v_i \in V(G)$, the time complexity of computing a keyword bit vector $v_i.BV$ is given by $O(|v_i.W|)$ (lines 1-3).
The computation of $v_i.NBV$ will cost $O(m\cdot deg)$, where $deg$ denotes the average degree of vertices in the data graph (lines 4-8).
The time complexity of computing the number, $v_i.nk$, of distinct neighbor keywords is given by $O(deg)$ (line 9).
Moreover, the cost of updating $v_i.Aux$ is $O(1)$ (line 10).
Thus, the time complexity of offline pre-computation is given by $O(|V(G)| \cdot deg\cdot m+\sum_{i=1}^{|V(G)|}|v_i.W|)$.

The index construction in Algorithm~\ref{alg:index2} includes local partitioning and cost calculation (as illustrated in Algorithm~\ref{alg:cm}), where the height of the tree index $\mathcal{I}$ is given by $\lceil log_{fanout} |V(G)| \rceil$. 
In particular, the time complexity of the cost-model-based graph partitioning for index nodes is given by $O(|V(G)|\cdot fanout \cdot m \cdot local\_iter \cdot global\_iter)$.
Moreover, the time complexity of the index construction is given by 
$O(|V(G)| \cdot fanout \cdot m \cdot local\_iter \cdot global\_iter \cdot \lceil log_{fanout} |V(G)| \rceil)$.

Overall, the offline pre-computation takes $O(|V(G)| \cdot m\cdot(deg + \lceil log_{fanout} |V(G)| \rceil \cdot fanout \cdot local\_iter \cdot global\_iter) + \sum_{i=1}^{|V(G)|}|v_i.W|)$.


\section{Online S$^3$AND Query Computation}
\label{sec:S3AND_query_processing}

In this section, we discuss in detail our online S$^3$AND query computation processing in Algorithm~\ref{alg:online}.

\subsection{Pruning for Index Nodes}
\label{sec: pruning for index}
In this subsection, we propose effective pruning methods on index nodes to prune index nodes with (a group of) vertex false alarms.

\noindent
\textbf{Index-Level Keyword Set Pruning:} 
If all vertices under an index entry $\mathcal{N}_i$ do not contain some keyword in $q_j.W$ for a query vertex $q_j\in V(q)$, then index entry $\mathcal{N}_i$ can be pruned w.r.t. this query vertex $q_j$ (i.e., $\mathcal{N}_i$ does not contain any vertices matching with $q_j$). 

Below, we provide an effective \textit{index-level keyword set pruning} method, using the $m$ aggregated keyword bit vectors $\mathcal{N}_i.BV^{(x)}$ stored in $\mathcal{N}_i$.

\begin{lemma}
    \label{lemma:index keyword set pruning}
    {\bf (Index-Level Keyword Set Pruning)} Given an index entry $\mathcal{N}_i$ and a query vertex $q_j\in V(q)$, index entry $\mathcal{N}_i$ can be pruned with respect to $q_j$, if $\bigvee_{x=1}^m \left(\mathcal{N}_i.BV^{(x)} \: \bigwedge \: q_j.BV^{(x)} \neq q_j.BV^{(x)}\right)$ holds.
\end{lemma}
\begin{proof}
    If $\bigvee_{x=1}^m \left(\mathcal{N}_i.BV^{(x)} \: \bigwedge \: q_j.BV^{(x)} \neq q_j.BV^{(x)}\right)$ holds, it means that all vertices in $\mathcal{N}_i$ do not contain some keyword in $q_j.W$. According to the constraint of the keyword set containment in Definition~\ref{def:S3AND}, index entry $\mathcal{N}_i$ cannot contain any candidate vertices matching with $q_j$. Thus, $\mathcal{N}_i$ can be safely pruned with respect to query vertex $q_j$.
\end{proof}


\noindent
\textbf{Index-Level ND-Lower-Bound-Based Pruning:} We also propose an index-level candidate node retrieval based on ND lower bounds (via Corollary~\ref{coro:ND_pruning_MAX_SUM}) below. 


\begin{lemma}
    \label{lemma:index_lb_ND_pruning}
    {\bf (Index-Level ND-Lower-Bound-Based Pruning)} Given an index entry $\mathcal{N}_i$, a query vertex $q_j$, and an aggregate threshold $\sigma$, index entry $\mathcal{N}_i$ can be safely pruned with respect to $q_j$, if $lb\_ND(q_j, \mathcal{N}_i) > \sigma$ holds, where we have $lb\_ND(q_j, \mathcal{N}_i)$  $=\min_{\forall v \in \mathcal{N}_i}$ $\{lb\_ND(q_j, v)\}$.
\end{lemma}
\begin{proof} (Proof by Contradiction) Assume that some vertex $v$ under index entry $\mathcal{N}_i$ is in an S$^3$AND subgraph answer $g$. Since it holds that $lb\_ND(q_j, \mathcal{N}_i)$ $=\min_{\forall v \in \mathcal{N}_i}\{lb\_ND(q_j, v)\}$, we have $lb\_ND(q_j, v) \geq lb\_ND(q_j, \mathcal{N}_i)$. Moreover, from the lemma assumption that $lb\_ND(q_j, \mathcal{N}_i) > \sigma$, by the inequality transition, we have $lb\_ND(q_j, v) > \sigma$. From Eq.~(\ref{eq:lb_AND}), for either $f=MAX$ or $f=SUM$, we always have $lb\_AND(q, g) = f(f(lb\_ND(q_j, v') | \forall v'\in V(g)\textbackslash\{v\}), lb\_ND(q_j, v)) > \sigma$, which contradicts with our initial assumption that vertex $v$ is in the subgraph answer $g$. Therefore, we can prune all vertices under entry $\mathcal{N}_i$ with respect to $q_j$.
\end{proof}



\noindent
\textbf{An Example of the Pruning via the Tree Index $\mathcal{I}$:} We use Example~\ref{ex:pruning} to illustrate the index pruning process.

\begin{example}
\label{ex:pruning}
    \textbf{(Pruning Over Index $\mathcal{I}$)} We continue with the example in Figures~\ref{fig:AND_application} and ~\ref{fig:indexShow} to illustrate the index pruning in Lemmas \ref{lemma:index keyword set pruning} and \ref{lemma:index_lb_ND_pruning}. For the index-level keyword set pruning, in Figure~\ref{fig:AND_application}, the keyword sets of vertices $v_7$ and $v_{10}$ have no intersection with that of any query vertex in query graph $q$. Since node $\mathcal{N}_{4}$ contains vertices $v_7$ and $v_{10}$, its aggregates $\mathcal{N}_{4}.BV^{(x)}$ satisfy the condition that: $\bigvee_{x=1}^m \left(\mathcal{N}_{4}.BV^{(x)} \: \bigwedge \: q_j.BV^{(x)} \neq q_j.BV^{(x)}\right)$ is true (assuming no conflicts in bit vectors), for each query vertex $q_j \in V(q)$. Thus, we can safely prune entry $\mathcal{N}_{4}$ in root $\mathcal{N}_0$ without accessing this branch.

    For the index-level ND-lower-bound-based pruning, in Figure~\ref{fig:indexShow}, the ND lower bound $lb\_ND(q_1, \mathcal{N}_1)$ for query vertex $q_1$ and node $\mathcal{N}_1$ is given by $\min_{\forall v \in \mathcal{N}_1} \{lb\_ND(q_1, v)\} = \min \{lb\_ND(q_1, v_1),$ $lb\_ND(q_1, v_{12})$. From Eqs.~(\ref{eq:lb_ND1}) and ~\ref{eq:lb_ND3}, we have the tight ND lower bounds $lb\_ND(q_1, v_1)=1$  and $lb\_ND(q_1, v_{12}) = 2$. Thus, we obtain $lb\_ND(q_1, \mathcal{N}_1) = \min \{1,2\} = 1$. Based on index-level ND-lower-bound-based pruning (Lemma \ref{lemma:index_lb_ND_pruning}), if the ND lower bound, $lb\_ND(q_1,\mathcal{N}_1)$ ($=1$), is greater than the threshold $\sigma$,
    we can safely prune entry $\mathcal{N}_1$ in root $\mathcal{N}_0$ without accessing the leftmost branch through pointer $\mathcal{N}_1.ptr$.

\end{example}

\nop{
\begin{lemma}
    \label{lemma:index_lb_ND_SUM_retrieval}
    {\bf (Index-Level ND-Lower-Bound-Based Node Retrieval (SUM Aggregate))} Given an index entry $\mathcal{N}_i$, a query vertex $q_j$, and a SUM aggregate threshold $\sigma_f$  ($=\sigma_{SUM} = \sigma/|V(q)|$), index entry $\mathcal{N}_i$ should be accessed with respect to $q_j$, if $lb\_ND(q_j, \mathcal{N}_i) \leq \sigma_{SUM}$ holds, where $lb\_ND(q_j, \mathcal{N}_i) = \min_{\forall v_i \in \mathcal{N}_i}\{lb\_ND(q_j, v_i)\}$.
    
\end{lemma}

{\color{red} {\bf [please update the proof correspondingly]}
\color{blue}
\begin{proof}
    If exists a vertex $v_i$ in the entry $\mathcal{N}_i$ satisfies $lb\_ND(q_j, v_i) < \sigma_f$ with the aggregate function $f$ being SUM, then according to the pigeonhole principle, it is possible that $lb\_AND(q,s) \leq \sigma$, where $\sigma = \sigma_f \cdot |V(q)|$. Therefore, it is necessary to retrieve the entry $\mathcal{N} $ to prevent the loss of final results.
\end{proof}

}

}

\begin{algorithm}[t]
    \caption{\bf Online S$^3$AND Query Processing}
    \label{alg:online}
    \KwIn{
        \romannumeral1) a data graph $G$,
        \romannumeral2) a query graph $q$,
        \romannumeral3) an aggregate threshold $\sigma$, 
        \romannumeral4) an aggregation function $f(\cdot)$, and
        \romannumeral5) the index $\mathcal{I}$ over $G$ 
    }
    \KwOut{a set, $S$, of subgraphs in $G$ similar to $q$ under AND semantics}

    \tcp{initialization}
    $S \leftarrow \emptyset$;

    \For{each query vertex $q_j\in V(q)$}
    {
        obtain $m$ keyword bit vectors $q_j.BV^{(x)}$ (for $1 \leq x \leq m$)
        
        $q_j.V_{cand} \leftarrow \emptyset$;
    }



    initialize a maximum heap $\mathcal{H}$ accepting entries in the form $(\mathcal{N}, key)$

    insert entry $(root(\mathcal{I}), 0)$ into heap $\mathcal{H}$

    $root(\mathcal{I}).Q = V(q)$;

    \tcp{index traversal}

    \While{$\mathcal{H}$ is not empty}{
        
        $(\mathcal{N}, key)$ $\leftarrow$ $\mathcal{H}.pop()$ 

        \eIf{$\mathcal{N}$ is a leaf node}{

            
            \For{each vertex $v_i \in \mathcal{N}$}{
                \For{each query vertex $q_j \in \mathcal{N}.Q$}
                {
                    \If{$v_i$ cannot be pruned by Lemma~\ref{lemma:keyword_pruning} and Corollary~\ref{coro:ND_pruning_MAX_SUM}}{
    

                        $q_j.V_{cand}  \leftarrow  q_j.V_{cand} \cup \{v_i\}$
            
                    }

                }
            
            }
            
        }{

            \tcp{$\mathcal{N}$ is a non-leaf node}
            \For{each entry $\mathcal{N}_i \in \mathcal{N}$}{

                $\mathcal{N}_i.Q \leftarrow \emptyset$;
                
                \For{each query vertex $q_j \in \mathcal{N}.Q$}
                {

                    \If{$\mathcal{N}_i$ cannot be pruned (w.r.t., $q_j$) by Lemmas~\ref{lemma:index keyword set pruning} and ~\ref{lemma:index_lb_ND_pruning}}{
                

                        $\mathcal{N}_i.Q \leftarrow \mathcal{N}_i.Q \cup \{q_j\}$
    
                    }

                }

                \If{$\mathcal{N}_i.Q$ is not empty}{
                    insert entry $(\mathcal{N}_i, \mathcal{N}_i.nk)$ into heap $\mathcal{H}$
                }                  
            
            }
            
        }
    }

    %


    refine candidate vertex sets $q_j.V_{cand}$ by checking the keyword matching with $q_j.W$ in query vertices $q_j$ (for $1\leq j\leq |V(q)|$)

    \tcp{generate a query plan $Q$}
    obtain the first query vertex $q_j$ with the smallest candidate vertex set $|q_j.V_{cand}|$, and initialize a sorted list (query plan) $Q = \{q_j\}$ 

    \While{$Q \neq V(q)$}{
        for all query vertices $q_i\in Q$, find a neighbor $q_l \in N(q_i)$ with the minimum candidate set size $|q_l.V_{cand}|$
        
        append $q_l$ to the end of the sorted list $Q$
    
    }

     \tcp{candidate subgraph retrieval and refinement}

    $S \leftarrow$ {\bf Refinement}$_f$ $(G, q, Q, S, \emptyset, 0, \sigma)$;
    


    


    


   
    
    





    \Return $S$
\end{algorithm}

\vspace{-1ex}
\subsection{S$^3$AND Query Algorithm}
Algorithm \ref{alg:online} illustrates the pseudo code of our proposed S$^3$AND query answering algorithm, which traverses the index $\mathcal{I}$ to retrieve candidate vertices (via pruning strategies) and refines candidate subgraphs by combining candidate vertices to return actual S$^3$AND answers. 

{\color{black}
\vspace{0.5ex}\noindent
\textbf{Initialization:}
When an S$^3$AND query arrives, we first initialize an empty S$^3$AND query answer set $S$ (containing subgraphs that satisfy both keyword and AND constraints w.r.t. query graph $q$; line 1). Moreover, for each query vertex $q_j \in V(q)$, we hash keywords in $q_j.W$ into $m$ keyword bit vectors $q_j.BV^{(x)}$, and initialize an empty set $q_j.V_{cand}$ to record candidate vertices that match with $q_j$ (lines 2-4).
We also maintain a \textit{maximum heap} $\mathcal{H}$ for the index traversal, accepting entries in the form $(\mathcal{N}, key)$, where $\mathcal{N}$ is an index entry, and $key$ is a heap entry key (defined as the $\mathcal{N}.nk$; intuitively, node entries with large keys tend to contain vertices with lower ND values; line 5). Then, we add the tree root $root(\mathcal{I})$  (in the form $(root(\mathcal{I}), 0)$) to $\mathcal{H}$, and let its corresponding query vertex set $root(\mathcal{I}).Q$ be $V(q)$ (lines 6-7). 
}



\vspace{0.5ex}\noindent
\textbf{Index Traversal:} We next traverse the index $\mathcal{I}$, by utilizing the maximum heap $\mathcal{H}$.
Each time, we pop out an entry $(\mathcal{N}, key)$ with the maximum key from $\mathcal{H}$ (lines 8-9). 
When $\mathcal{N}$ is a leaf node, we will check each vertex $v_i \in \mathcal{N}$. That is, with respect to each query vertex $q_j \in \mathcal{N}.Q$, if vertex $v_i$ cannot be pruned by the \textit{Keyword Set} and \textit{ND Lower Bound Pruning} (Lemma~\ref{lemma:keyword_pruning} and Corollary~\ref{coro:ND_pruning_MAX_SUM}, respectively), then we add $v_i$ to candidate vertex set $q_j.V_{cand}$ of $q_j$ (lines 10-14).


When $\mathcal{N}$ is a non-leaf node, we consider each node entry $\mathcal{N}_i \in \mathcal{N}$ and check whether we need to access the children of entry $\mathcal{N}_i$ (lines 15-22). In particular, we first initialize an empty query set $\mathcal{N}_i.Q$, and then check if entry $\mathcal{N}_i$ can be pruned with respect to each query vertex $q_j \in \mathcal{N}.Q$ by Lemmas~\ref{lemma:index keyword set pruning} and \ref{lemma:index_lb_ND_pruning}. If $\mathcal{N}_i$ cannot be ruled out (w.r.t. $q_j$), we will add $q_j$ to $\mathcal{N}_i.Q$ (lines 17-20). In the case that $\mathcal{N}_i.Q$ is not empty, we insert entry $(\mathcal{N}_i, \mathcal{N}_i.nk)$ into heap $\mathcal{H}$ for later investigation (lines 21-22).


\vspace{0.5ex}\noindent
\textbf{Candidate Subgraph Retrieval and Refinement:} After the index traversal, we obtain a candidate vertex set $q_j.V_{cand}$ for each query vertex $q_j$. Since we used keyword bit vectors for pruning, there may still exist some false positives. We thus need to refine candidate vertices $v_i$ in $q_j.V_{cand}$, by comparing the actual keyword sets (i.e., checking $q_j.W \subseteq v_i.W$; line 23). 

Then, we will compute a query plan $Q$, which is a sorted list of query vertices in $q$ to guide the order of candidate vertex concatenation and obtain candidate subgraphs (lines 24-27). Specifically, we initialize the first vertex in $Q$ with a query vertex $q_j$ with the smallest candidate set size $|q_j.V_{cand}|$ (line 24). Next, each time we append a query vertex $q_l \in V(q)$ to the end of the sorted list $Q$, until $Q=V(q)$ holds, where $q_l$ is a neighbor of $q_i$ (for all $q_i\in Q$) with the minimum candidate set size $|q_l.V_{cand}|$ (lines 25-27). After that, we call function {\bf Refinement}$_f$ $(G,q,Q,S,\emptyset,0,\sigma)$ in Algorithm~\ref{alg:css}, and return final S$^3$AND subgraph answers in $S$ (lines 28-29).

\begin{algorithm}[t]
    \caption{\bf Refinement$_{f}$}
    \label{alg:css}
    \KwIn{
        \romannumeral1) a data graph $G$,
        \romannumeral2) a query graph $q$,
        \romannumeral3) a sorted candidate vertex list (query plan) $Q$,
        \romannumeral4) a vertex list, $M$, matching with query vertices in $Q$,
        \romannumeral5) a recursion depth $dep$, and
        \romannumeral6) an aggregate threshold $\sigma$
    }
    \KwOut{a set, $S$, of subgraphs that satisfy keyword and AND constraints for $q$}

    \eIf{$|Q|=dep$}{

        \If{subgraph $g$ with vertices $V(g) = M$ is connected and $AND(q, g) \leq \sigma$}{
        
            $S \leftarrow S \cup \{g\}$
            
        }
        
    
    }{



            

            \For{each candidate vertex $v \in Q[dep].V_{cand}$ and $v \notin M$}{


                \If{vertex $v$ is connected to some vertex $M[i]$ (for $0 \leq i < dep$) or some candidate vertex in $Q[i].V_{cand}$ (for $dep+1 \leq i <  |Q|$)}{

                    $M[dep] = v$

                    {\bf Refinement}$_{f}$ $(G, q, Q, S, M, dep+1, \sigma)$

                }


                    

                }

            
    
    }

    \Return $S$
    


    
    
\end{algorithm}

\vspace{0.5ex}\noindent \textbf{Discussions on How to Retrieve and Refine Candidate Subgraphs:} Algorithm~\ref{alg:css} illustrates the pseudo code of a recursive function to retrieve and refine candidate subgraphs. Specifically, in the base case that the recursive depth $dep$ is $|Q|$, we have all the matching pairs of vertices between $M$ and $Q$, and check the S$^3$AND constraints between subgraph $g$ (with vertices in $M$) and query graph $q$. If candidate subgraph $g$ is the S$^3$AND query answer, then we add $g$ to the answer set $S$ (lines 1-3). 

When the recursive depth $dep$ has not reached $|Q|$, we will consider each candidate vertex $v$ in $Q[dep].V_{cand}$ (not a duplicate in $M$; line 5). In particular, if candidate vertex $v$ is not connected to some vertex $M[i]$ ($0\leq i < dep$) in the current vertex list $M$ and any vertex in $Q[i].V_{cand}$ ($dep+1\leq i <|Q|$), then it implies that the resulting subgraph $g$ will not be connected and we can terminate the recursive call; otherwise, we can set the matching vertex $M[dep]$ to $v$, and recursively invoke function {\bf Refinement}$_f (G,q,Q,S,M,dep+1,\sigma)$ for the next depth $(dep+1)$ (lines 6-8). Finally, we return the S$^3$AND query answer set $S$ (line 9).

\vspace{0.5ex}\noindent
\textbf{Complexity Analysis: }
In Algorithm~\ref{alg:online}, the time complexity of the initialization is $O(\sum_{j=1}^{|V(q)|} |q_j.W|)$. Let $PP_i$ denotes the pruning power (i.e., the percentage of node entries that can be pruned) on the $i$-th level of the tree index $\mathcal{I}$, where $1 \leq i \leq \lceil log_{fanout} |V(G)| \rceil$.
Then, for the index traversal process, the number of visited nodes is $\sum_{i=1}^{\lceil log_{fanout} |V(G)| \rceil} fanout \cdot (1-PP_i)$.
The update cost of each candidate vertex set is $O(1)$. In the refinement process, the time complexity of generating the query plan $Q$ is $O(|V(q)|)$. Since we adopt a recursive strategy to handle candidate nodes in the query plan, with a recursion depth of $|Q|$ and an update cost of $ O(1)$ each time,  the worst-case time complexity is $O(\prod_{i=1}^{|Q|} |Q[i].V_{cand}|)$.
 
Therefore, the overall time complexity of online S$^3$AND query processing (i.e.,  Algorithm~\ref{alg:online}) is given by $O\big(\sum_{j=1}^{|V(q)|} |q_j.W| + \\ \sum_{i=1}^{\lceil log_{fanout} |V(G)| \rceil}  fanout \cdot (1-PP_i)+|V(q)| +\prod_{i=1}^{|Q|}|Q[i].V_{cand}|\big)$.

\section{Experimental Evaluation}
\label{sec:exper}

In this section, we evaluate the performance of our proposed S$^3$AND approach (i.e., Algorithm~\ref{alg:online}) on real/synthetic graphs.

\subsection{Experimental Settings}
\label{sec:experimental setting}
 
\textbf{Real-World Graph Data Sets:} 
We use 5 real-world graphs, $Facebook$ \cite{yang2023pane}, $PubMed$ \cite{meng2019co}, $Elliptic$ \cite{weber2019anti}, $TWeibo$ \cite{yang2023pane}, and $DBLPv14$ \cite{tang2008arnetminer}, whose statistics are depicted in Table~\ref{tab:datasets}, where ``Abbr.'' stands for the abbreviation of the name, and ``$|\sum|$'' stands for the keyword domain size.
$Facebook$ is a social network, where two users are connected if they are friends, and keywords of each user are obtained from one's profile.
$PubMed$ is a citation network of scientific publications on diabetes, where keywords are from lexical features.
$Elliptic$ is a Bitcoin transaction network, where each node represents a transaction, edges represent financial connections, and each node contains transaction attribute category keywords.
$TWeibo$ is also a social network, where each node represents a user, each edge represents a following relationship, and the keyword for each node is from the user profile.
$DBLPv14$ is a citation network extracted from DBLP, where each author's keywords were extracted from their relevant paper titles.
It is worth noting that if a vertex in the graph mentioned above does not have keyword attributes, we create a dummy keyword label ``0'' to indicate \textit{None} for this vertex.




\begin{table}[t]
\begin{center}
\caption{Statistics of the Tested Real-World Graph Data Sets.}
\label{tab:datasets}\vspace{-2ex}
\footnotesize
\begin{tabular}{|l||l||l||l||l|}
\hline
\textbf{Name} & \textbf{Abbr.} & $\bm{|V(G)|}$ & $\bm{|E(G)|}$ & $\bm{|\sum|}$ \\
\hline\hline
    $Facebook$~\cite{yang2023pane} & FB & 4,039 & 88,234 & 1,284 \\\hline
    $PubMed$~\cite{meng2019co} & PM & 19,717 & 44,338 & 501 \\\hline
    $Elliptic$~\cite{weber2019anti} & EL & 203,769 & 234,355 & 166 \\\hline
    $TWeibo$~\cite{yang2023pane} & TW & 2,320,895 & 9,840,066 & 1,658 \\\hline
    $DBLPv14$~\cite{tang2008arnetminer} & DB & 2,956,012 & 29,560,025 & 7,990,611 \\\hline
\end{tabular}
\end{center}
\end{table}

\begin{table}[t]
\begin{center}
\caption{Parameter Settings.}
\label{tab:parameters}\vspace{-2ex}
\resizebox{\linewidth}{!}{
\begin{tabular}{|p{5cm}||l|}
\hline
\textbf{Parameters}&\textbf{Values} \\
\hline\hline
    the threshold, $\sigma_{MAX}$ ($=$ $\sigma$), of MAX neighbor difference & {\bf1}, 2, 3, 4 \\\hline
    the threshold, $\sigma_{SUM}$ ($=$ $\sigma$), of SUM neighbor difference & 2, {\bf 3}, 4, 5 \\\hline
    the size, $|v_i.W|$, of keywords per vertex & 1, 2, {\bf 3}, 4, 5 \\\hline
    the size, $|\sum|$, of the keyword domain & 10, 20, {\bf50}, 80 \\\hline
    the size, $|V(q)|$, of query graph $q$ & 3, {\bf 5}, 8, 10 \\\hline
    the size, $|V(G)|$, of data graph $G$ & 10K, 25K, {\bf 50K}, 100K, 250K, 1M, 10M, 30M\\\hline    
\end{tabular}
}
\end{center}
\vspace{2ex}
\end{table}

\noindent
\textbf{Synthetic Graph Data Sets:}
We generate synthetic small-world graphs in the Newman-Watts-Strogatz model~\cite{watts1998collective}, using NetworkX \cite{hagberg2008exploring}, where parameters are depicted in Table~\ref{tab:parameters}.
By using different distributions of keywords in $\sum$ (i.e., \textit{Uniform}, \textit{Gaussian}, and \textit{Zipf}), we obtain three types of synthetic graphs: \textit{Syn-Uni}, \textit{Syn-Gau}, and \textit{Syn-Zipf}, respectively. 

\noindent
\textbf{Query Graph:} For each graph data set $G$, we randomly sample 100 connected subgraphs. For each extracted subgraph $g$, we remove each of its edges with probability 0.3 (as long as the subgraph $g$ is connected after the edge removal). As a result, we obtain 100 query graphs $q$ for the S$^3$AND query evaluation.


\noindent
\textbf{Comparison Methods:}
To our best knowledge, no prior work studied the S$^3$AND problem. Thus, we will compare our S$^3$AND approach (i.e., Algorithm~\ref{alg:online}) with three methods, a straightforward baseline method (named $Baseline$), $CSI\_GED$~\cite{gouda2016csi_ged}, and $MCSPLIT$~\cite{mccreesh2017partitioning}.
Specifically, for each vertex in a query graph $q$, the $Baseline$ method obtains a set of candidate vertices in the data graph that match with keywords in $q$. Then, we aggregate those subgraphs from $|V(q)|$ candidate vertex sets that meet the S$^3$AND query requirements as the returned results. 
$CSI\_GED$~\cite{gouda2016csi_ged} and $MCSPLIT$~\cite{mccreesh2017partitioning} first retrieve a set of candidate subgraphs that are similar to a given query graph $q$, where the graph similarity is measured by the \textit{Graph Edit Distance} (GED) or \textit{Maximum Common Subgraph} (MCS), respectively. After obtaining $CSI\_GED$ (or $MCSPLIT$) candidate subgraphs, we next refine/return subgraphs (in these candidates) so that they satisfy the keyword set constraints and have the same size as the query graph (note: the subgraph connectivity is relaxed and not required). In the experiments, we set the parameters to default values, and let the GED (or MCS) threshold be 1 for $CSI\_GED$ (or $MCSPLIT$) by default.

\noindent
\textbf{Measures:}
We evaluate the S$^3$AND query performance, in terms of \textit{pruning power} and \textit{wall clock time}. The \textit{pruning power} is the percentage of candidate vertices pruned by our pruning strategies, whereas the \textit{wall clock time} is the average time cost to answer S$^3$AND queries. 
We report the average values of the evaluated metrics over 100 runs.

\noindent
\textbf{Parameter Settings:} 
Table~\ref{tab:parameters} shows the parameter settings, where default values are in bold. 
Each time we vary one parameter, while setting other parameters to default values.
By default, we set the keyword group number $m$ to 5, and the fanout, $fanout$, of index $\mathcal{I}$ to 16. 
For the index construction, $global\_iter$ and $local\_iter$ are set to 5 and 20, respectively.
We ran all the experiments on a PC with 
an AMD Ryzen Threadripper 3990X CPU, 256 GB of memory, and 128 threads.
All algorithms were implemented in Python and executed with Python 3.11 interpreter.

\begin{figure}[t]
    \centering\vspace{-1ex}
    \subfigure[real-world graphs]{
            \includegraphics[width=0.48\linewidth]{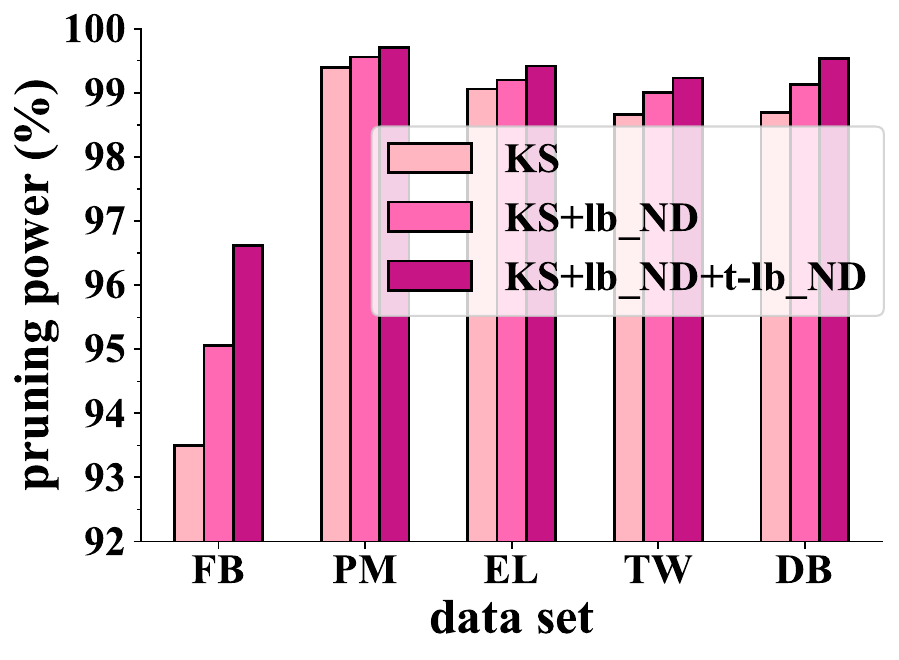}
            \label{subfig:S3AND_real_pruning}
    } \hspace{-0.2cm}
    \subfigure[synthetic graphs]{
        \includegraphics[width=0.48\linewidth]{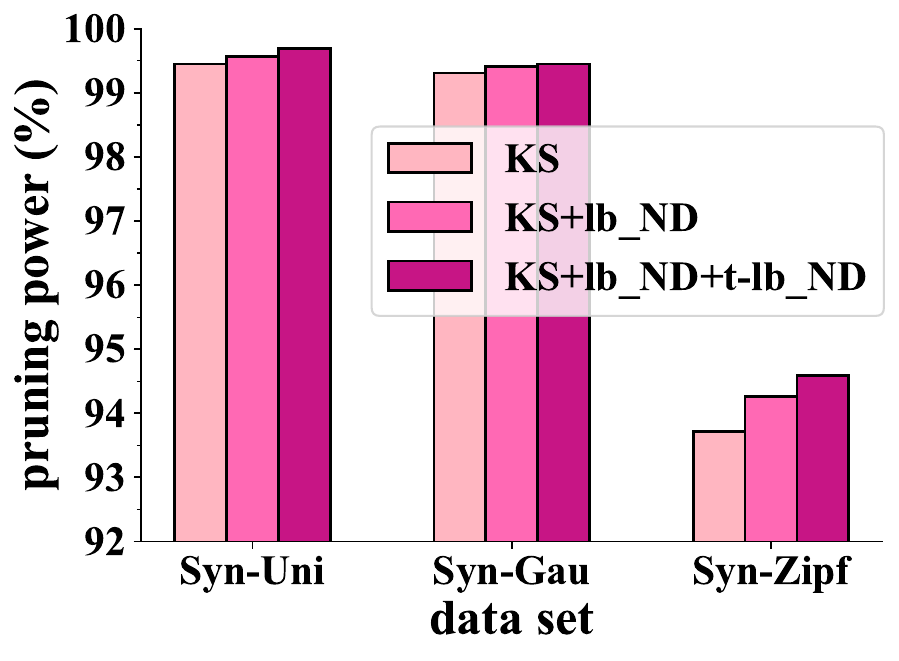}
        \label{subfig:S3AND_synthetic_pruning}
    }\vspace{-2ex}
    \caption{Effectiveness evaluation of S$^3$AND.}
    \label{fig:pruning}
    \vspace{-1ex}
\end{figure}


\begin{figure}[t]
    \centering
    \includegraphics[width=0.98\linewidth]{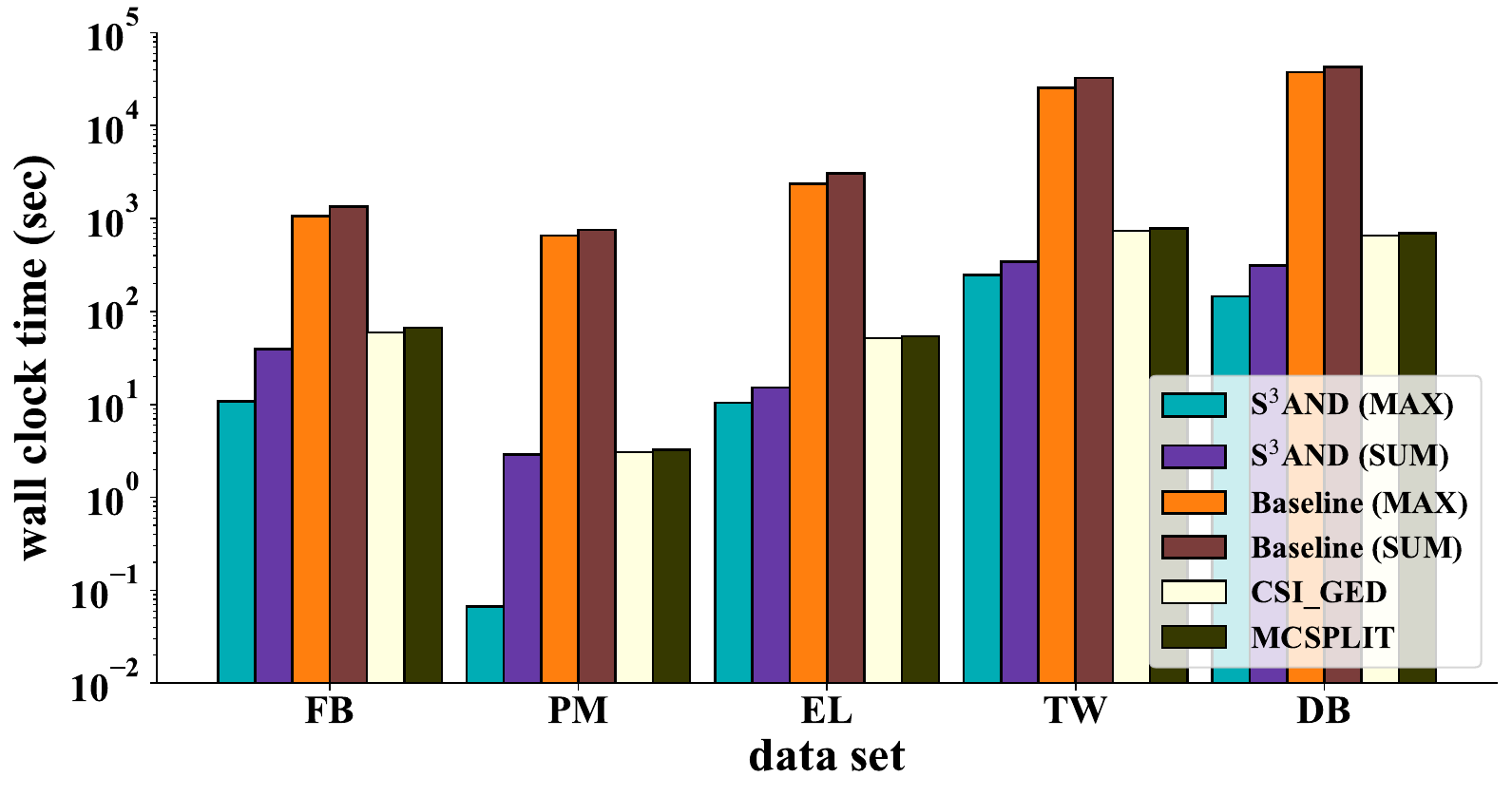}
    \caption{The comparison of our S$^3$AND approach with the \textit{Baseline}, CSI\_GED and MCSPLIT methods over real graphs.}
    \label{fig:baseline}
\end{figure}

\subsection{The S$^3$AND Effectiveness Evaluation}
\label{sec:evaluation}
In this subsection, we report the pruning power of our proposed pruning strategies in Section~\ref{sec: pruning for index} for S$^3$AND query processing over real-world/synthetic graphs. 

\begin{figure*}
    \centering
    \subfigure[$Syn\text{-}Uni$ vs. $\sigma_{MAX}$]{
            \includegraphics[width=0.15\linewidth]{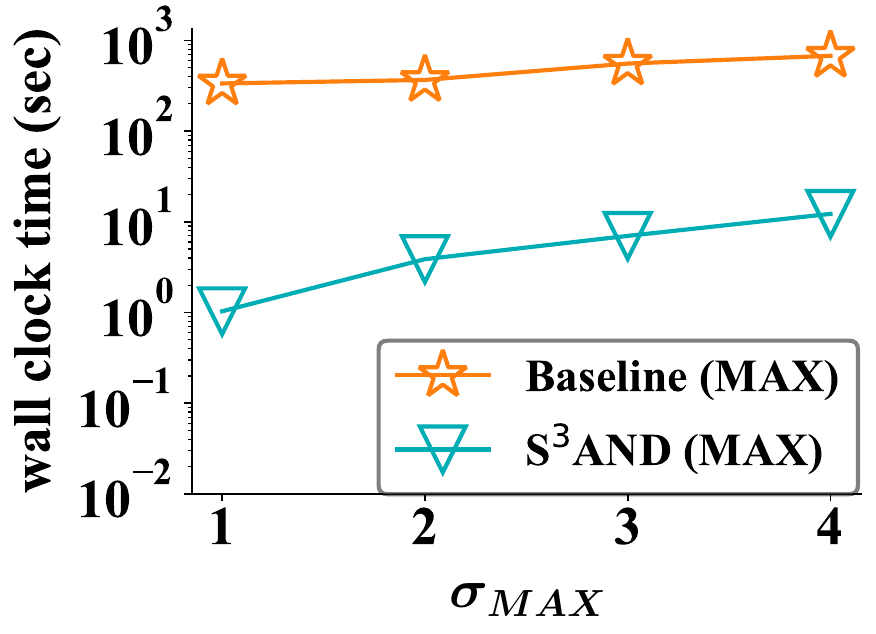}
            \label{subfig:uni-max}
    } \hspace{-0.2cm}
    \subfigure[$Syn\text{-}Gau$ vs. $\sigma_{MAX}$]{
        \includegraphics[width=0.15\linewidth]{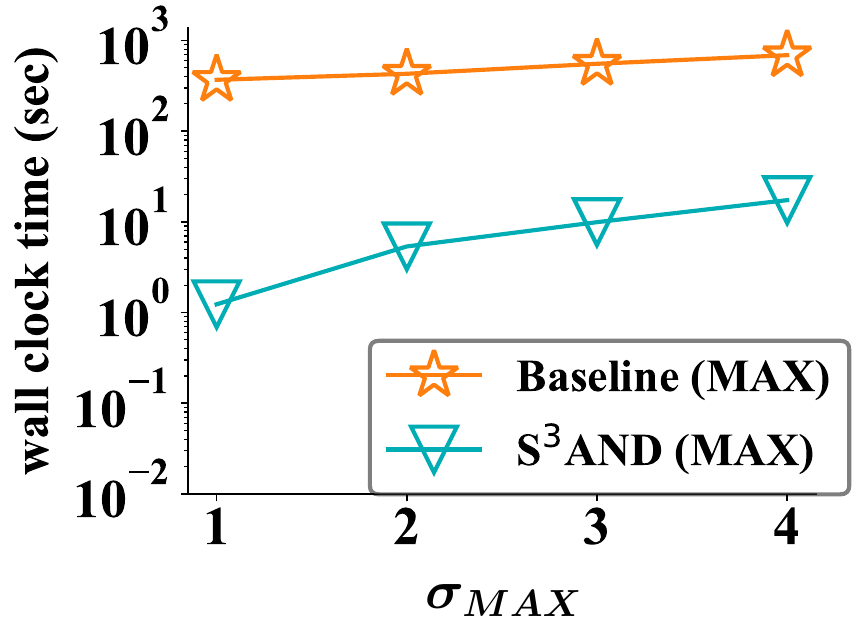}
        \label{subfig:gau-max}
    } \hspace{-0.2cm}
    \subfigure[$Syn\text{-}Zipf$ vs. $\sigma_{MAX}$]{
        \includegraphics[width=0.15\linewidth]{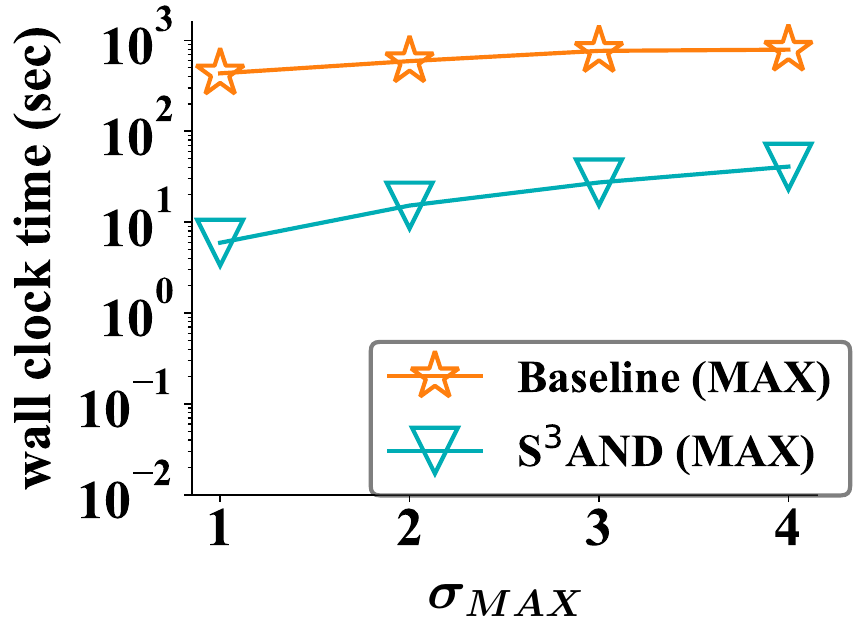}
        \label{subfig:zipf-max}
    } \hspace{-0.2cm}
    \subfigure[$Syn\text{-}Uni$ vs. $\sigma_{SUM}$]{
            \includegraphics[width=0.15\linewidth]{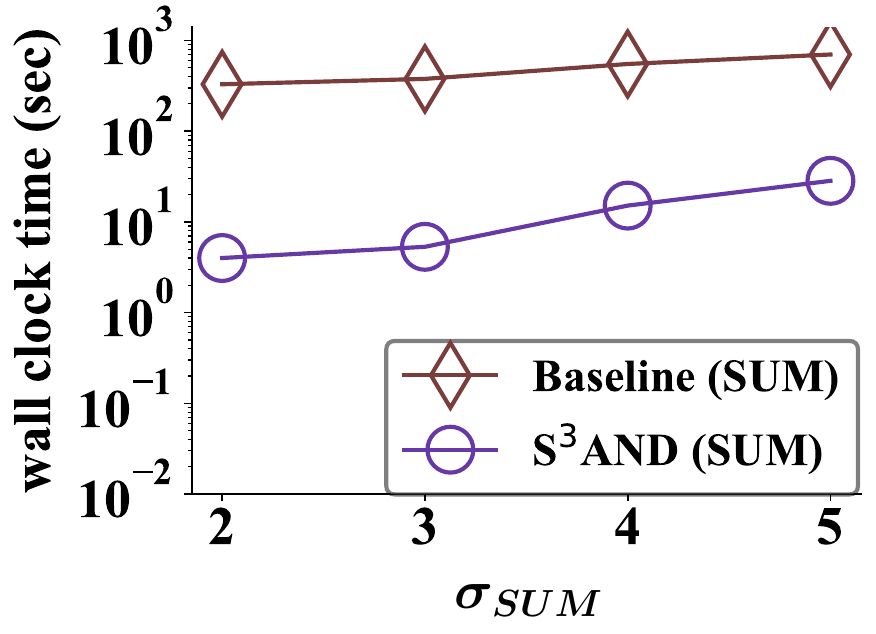}
            \label{subfig:uni-sum}
    } \hspace{-0.2cm}
    \subfigure[$Syn\text{-}Gau$ vs. $\sigma_{SUM}$]{
        \includegraphics[width=0.15\linewidth]{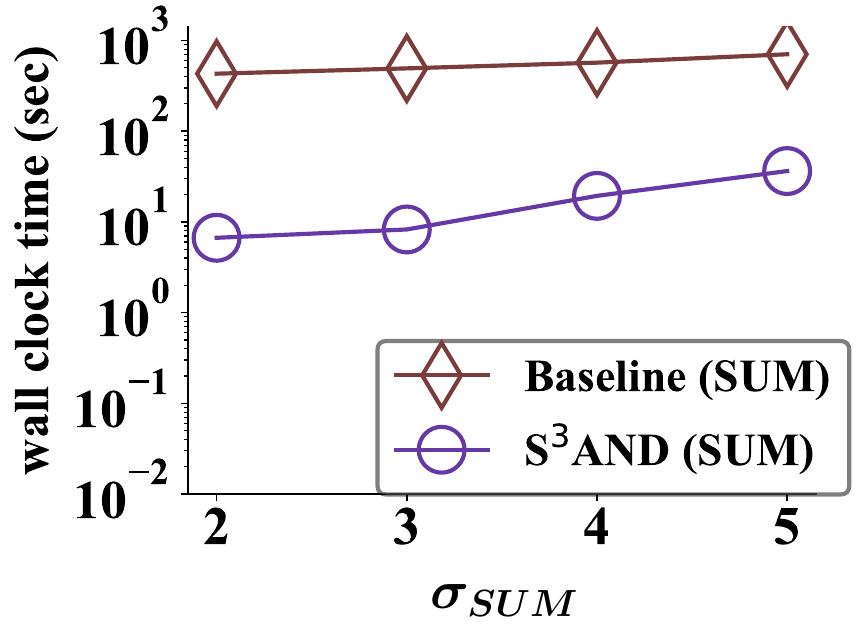}
        \label{subfig:gau-sum}
    } \hspace{-0.2cm}
    \subfigure[$Syn\text{-}Zipf$ vs. $\sigma_{SUM}$]{
        \includegraphics[width=0.15\linewidth]{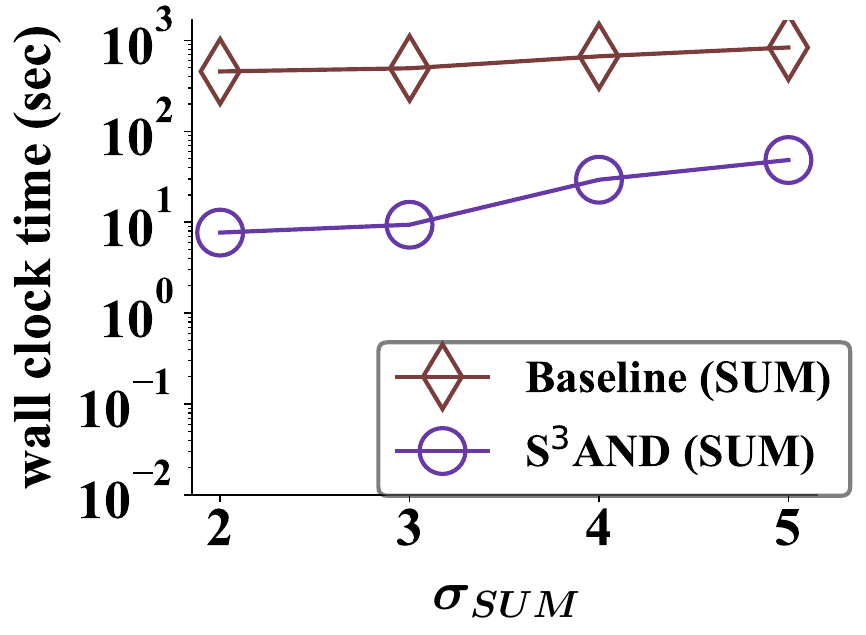}
        \label{subfig:zipf-sum}
    } \\
    \subfigure[$Syn\text{-}Uni$ vs. $|v_i.W|$]{
            \includegraphics[width=0.15\linewidth]{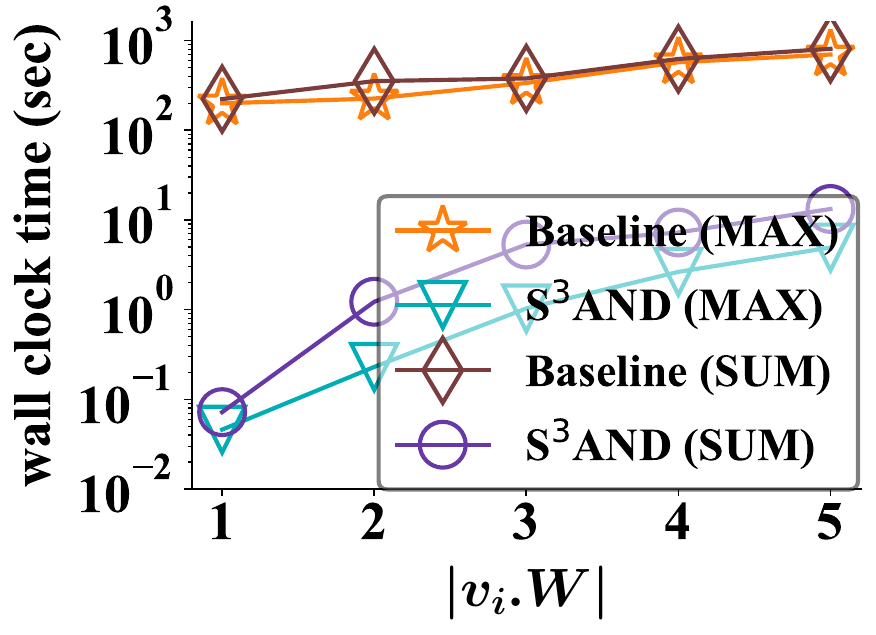}
            \label{subfig:uni-w}
    } \hspace{-0.2cm}
    \subfigure[$Syn\text{-}Gau$ vs. $|v_i.W|$]{
        \includegraphics[width=0.15\linewidth]{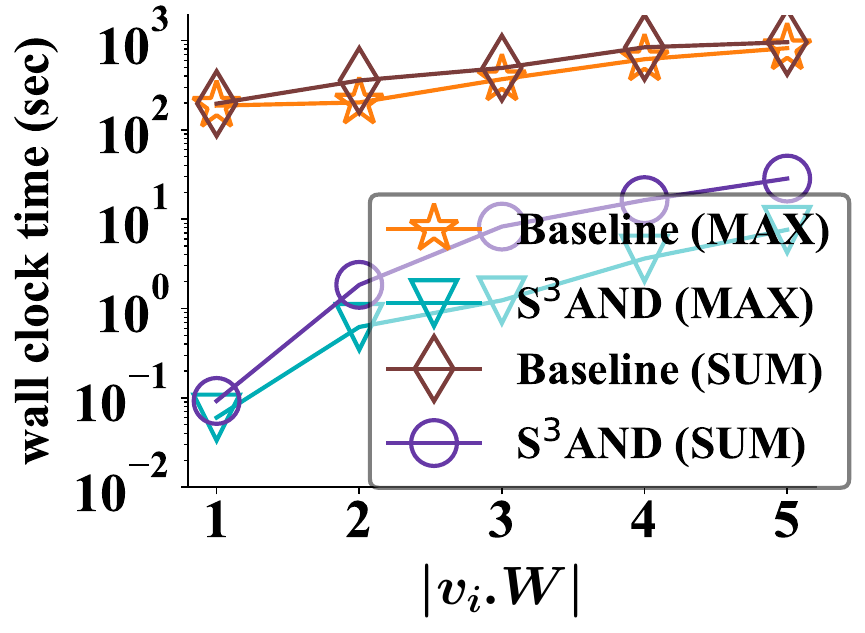}
        \label{subfig:gau-w}
    } \hspace{-0.2cm}
    \subfigure[$Syn\text{-}Zipf$ vs. $|v_i.W|$]{
        \includegraphics[width=0.15\linewidth]{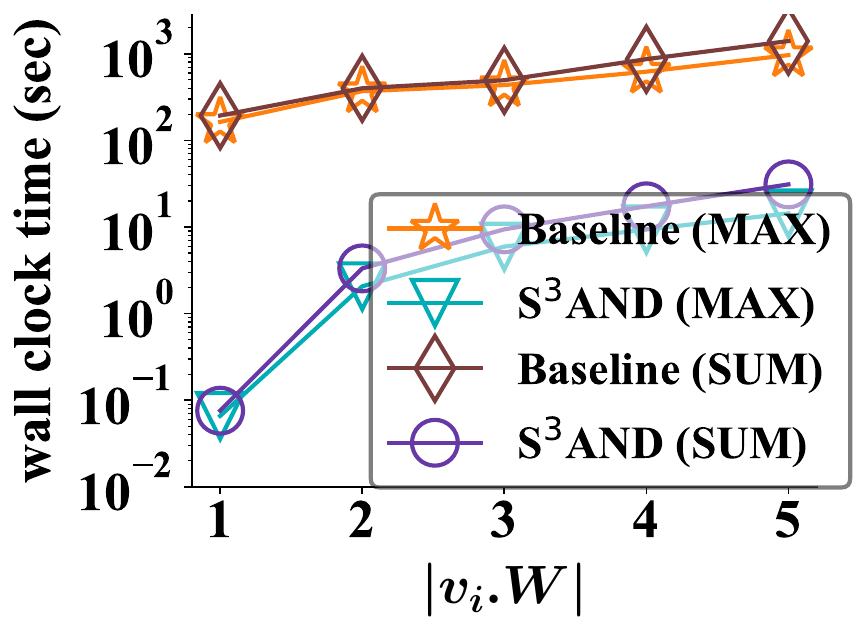}
        \label{subfig:zipf-w}
    } \hspace{-0.2cm}
    \subfigure[$Syn\text{-}Uni$ vs. $|\sum|$]{
            \includegraphics[width=0.15\linewidth]{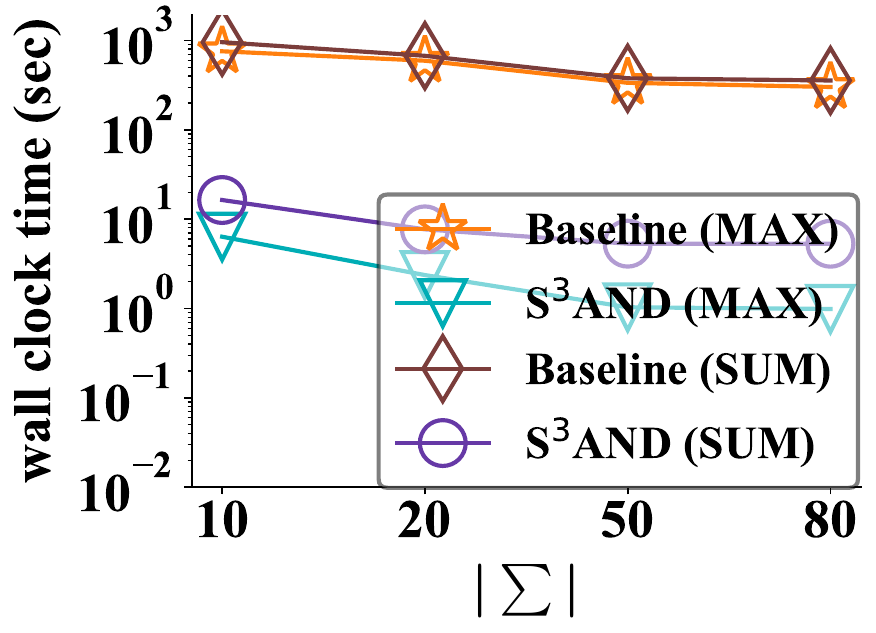}
            \label{subfig:uni-domain}
    } \hspace{-0.2cm}
    \subfigure[$Syn\text{-}Gau$ vs. $|\sum|$]{
        \includegraphics[width=0.15\linewidth]{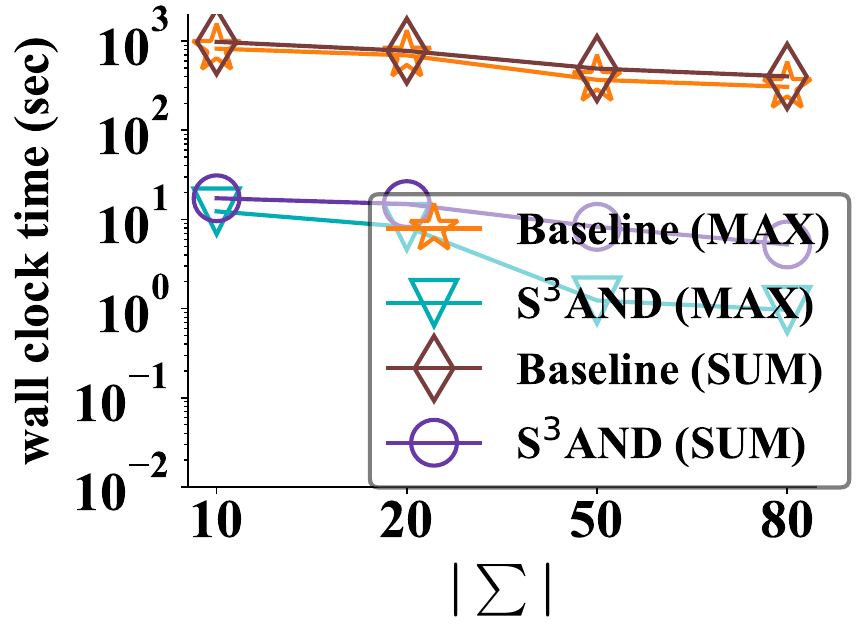}
        \label{subfig:gau-domain}
    } \hspace{-0.2cm}
    \subfigure[$Syn\text{-}Zipf$ vs. $|\sum|$]{
        \includegraphics[width=0.15\linewidth]{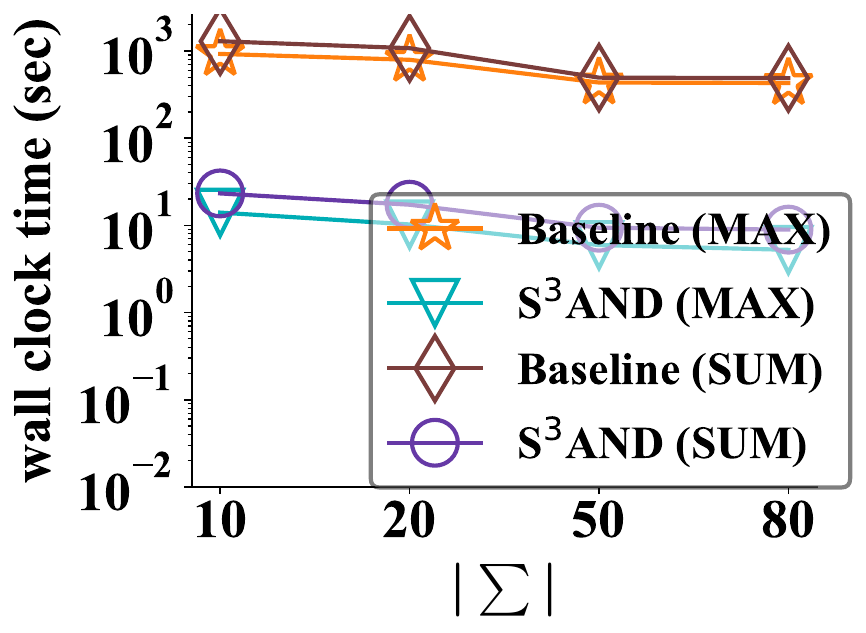}
        \label{subfig:zipf-domain}
    } \\
    \subfigure[$Syn\text{-}Uni$ vs. $|V(q)|$]{
            \includegraphics[width=0.15\linewidth]{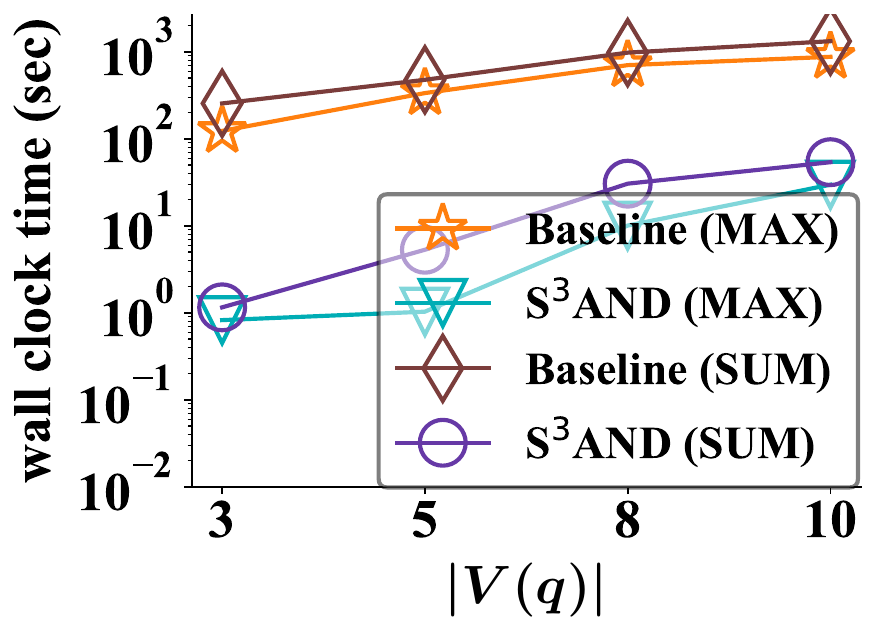}
            \label{subfig:uni-vq}
    } \hspace{-0.2cm}
    \subfigure[$Syn\text{-}Gau$ vs. $|V(q)|$]{
        \includegraphics[width=0.15\linewidth]{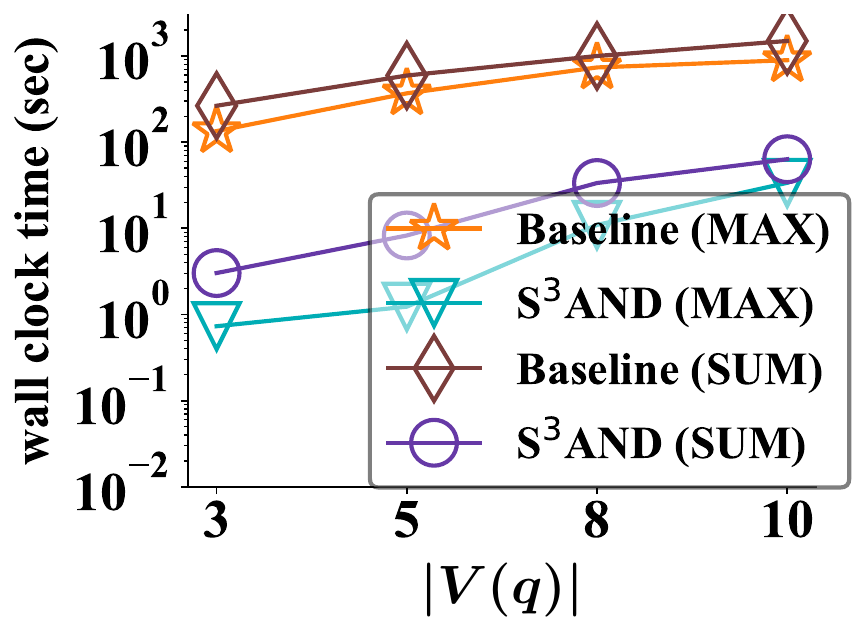}
        \label{subfig:gau-vq}
    } \hspace{-0.2cm}
    \subfigure[$Syn\text{-}Zipf$ vs. $|V(q)|$]{
        \includegraphics[width=0.15\linewidth]{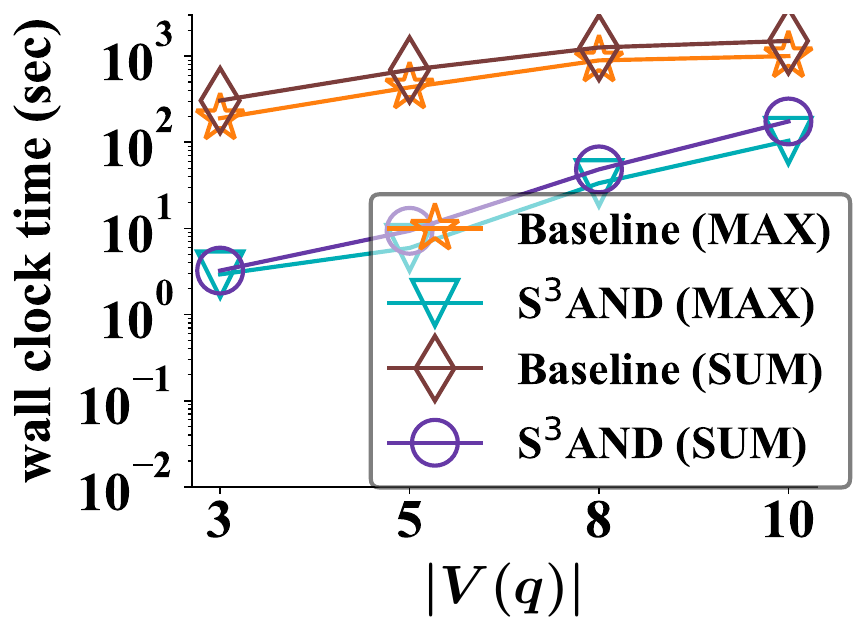}
        \label{subfig:zipf-vq}
    } \hspace{-0.2cm}
    \subfigure[$Syn\text{-}Uni$ vs. $|V(G)|$]{
            \includegraphics[width=0.15\linewidth]{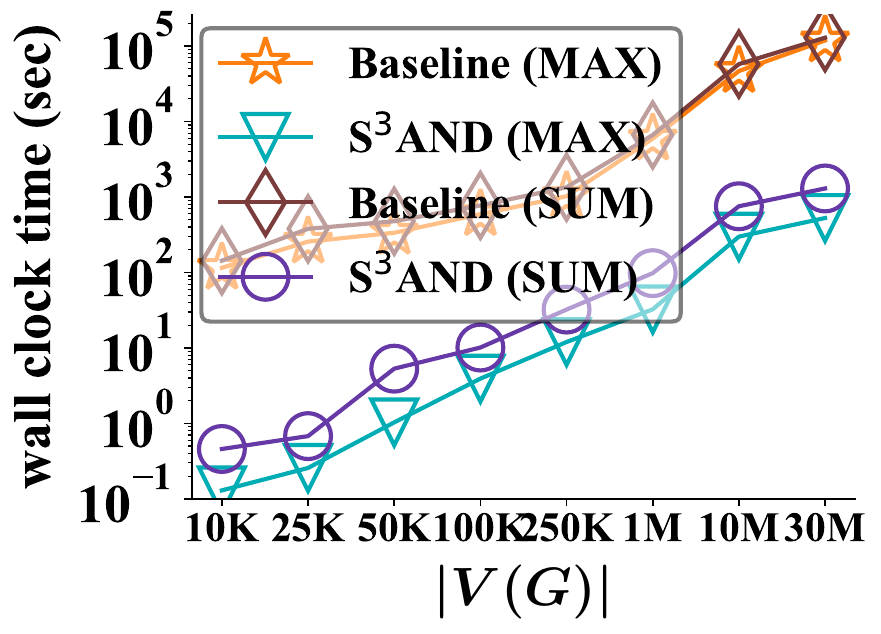}
            \label{subfig:uni-vG}
    } \hspace{-0.2cm}
    \subfigure[$Syn\text{-}Gau$ vs. $|V(G)|$]{
        \includegraphics[width=0.15\linewidth]{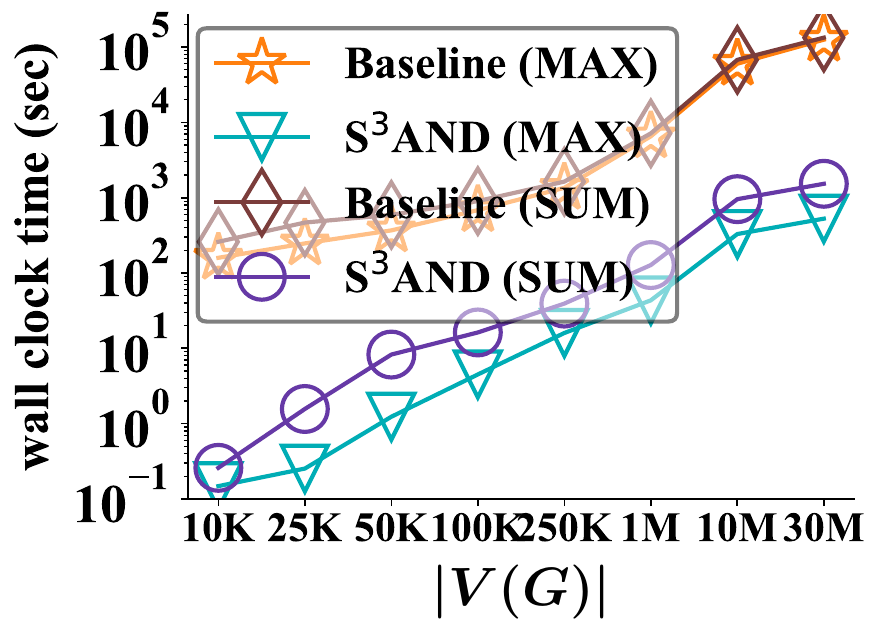}
        \label{subfig:gau-vG}
    } \hspace{-0.2cm}
    \subfigure[$Syn\text{-}Zipf$ vs. $|V(G)|$]{
        \includegraphics[width=0.15\linewidth]{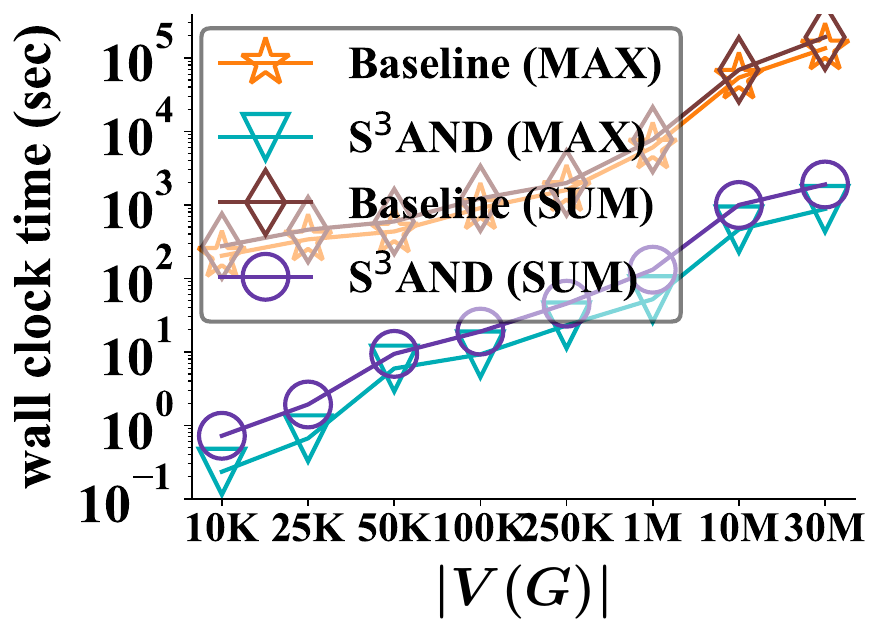}
        \label{subfig:zipf-vG}
    } 
    \vspace{-2ex}
    \caption{The S$^3$AND query efficiency on synthetic graphs, compared with the $Baseline$ method.}
    \label{fig:online_baseline}
    \vspace{-3ex}
\end{figure*}

Figure~\ref{fig:pruning} conducts an ablation study on the pruning power of different pruning combinations in our S$^3$AND approach over real-world/synthetic graphs. Each time we add one more pruning strategy and test three pruning combinations: (1) \textit{keyword set pruning} ($KS$), (2) \textit{keyword set + ND lower bound pruning} ($KS + lb\_ND$; note: $lb\_ND(\cdot)$ is given by Eq.~(\ref{eq:lb_ND1})), (3) \textit{keyword set + ND lower bound pruning + Tighter ND lower bound pruning} ($KS + lb\_ND + t\text{-}lb\_ND$; note: $t\text{-}lb\_ND(\cdot)$ is given by Eq.~(\ref{eq:lb_ND3})), where all parameters are set to default values. From the figures, we can see that our proposed pruning combinations can achieve high pruning powers (i.e., above 93.5\%) for both real and synthetic graphs. With more pruning strategies used, our proposed S$^3$AND approach can have higher pruning power. The overall pruning power with all the three pruning methods can reach $96.62\% \sim 99.70\%$ for real-world graphs and $94.59\% \sim 99.69\%$ for synthetic graphs, which confirms the effectiveness of our proposed pruning strategies.


\subsection{The S$^3$AND Efficiency Evaluation}
\label{sub:s3and}
In this subsection, we compare our online S$^3$AND query algorithm with $Baseline$, $CSI\_GED$, and $MCSPLIT$, under default parameters over real-world graphs, in terms of the wall clock time. Figure~\ref{fig:baseline} illustrates the comparative results over real-world graphs, where parameters are set to their default values. 
From the figure, we can see that the S$^3$AND query efficiency outperforms that of $Baseline$ by 1-3 orders of magnitude, for either SUM or MAX aggregate.
For example, when $f$ is MAX, our S$^3$AND query time is $0.07 \sim 246.32$ $sec$ for real-world graphs and $1.03\sim 5.94$ $sec$ for synthetic graphs. Moreover, our S$^3$AND algorithm incurs lower time cost than $CSI\_GED$ and $MCSPLIT$ methods, due to the costly calculation of GED or MCS. Note that, since $CSI\_GED$ and $MCSPLIT$ have different semantics from S$^3$AND and are used as the filter to retrieve candidate subgraphs, $CSI\_GED$ and $MCSPLIT$ may not return subgraph answers of good quality (please refer to the case study in Figure \ref{fig:case} of Section \ref{subsec:case study}).

To verify the robustness of our S$^3$AND approach, in the sequel, we will compare with $Baseline$ (which also outputs exact S$^3$AND answers) and test different parameters (e.g., $\sigma$, $|v_i.W|$, $|\sum|$, $|V(q)|$, and $|V(G)|$) on synthetic graphs (i.e., $Syn\text{-}Uni$, $Syn\text{-}Gau$, and $Syn\text{-}Zipf$).


\noindent
\textbf{The Efficiency w.r.t. the Threshold, $\sigma_{MAX}$, of MAX Neighbor Difference:} 
Figures~\ref{subfig:uni-max}, ~\ref{subfig:gau-max}, and ~\ref{subfig:zipf-max} illustrate the S$^3$AND query performance for MAX aggregate (i.e., $f=MAX$), compared with $Baseline$, where the AND threshold $\sigma_{MAX}$ varies from 1 to 4, and other parameters are set to default values. 
From the figures, we can see that for both S$^3$AND and $Baseline$, the wall clock time increases for larger $\sigma_{MAX}$ over all three synthetic graphs. This is because a larger MAX threshold $\sigma_{MAX}$ results in more candidate vertices, thereby raising the refinement cost. 
Nevertheless, our S$^3$AND approach outperforms $Baseline$ by 1-3 orders of magnitude, and remains low (i.e., $1.03 \sim 40.83 \: sec$) over three synthetic graphs.


\noindent
\textbf{The Efficiency w.r.t. the Threshold, $\sigma_{SUM}$, of SUM Neighbor Difference:} 
Figures~\ref{subfig:uni-sum}, ~\ref{subfig:gau-sum}, and ~\ref{subfig:zipf-sum} compare the S$^3$AND query performance for SUM aggregate (i.e., $f=SUM$) with that of $Baseline$, where $\sigma_{SUM} = 2, 3, 4,$ and $5$, and other parameters are by default. 
Similar to MAX aggregate threshold $\sigma_{MAX}$, when $\sigma_{SUM}$ increases, more candidate vertices will be retrieved for the refinement, which leads to higher query processing cost.
Nonetheless, for all the three synthetic graphs, our S$^3$AND approach takes $3.99 \sim 48.36 \: sec$ query time, and performs significantly better than $Baseline$ by about 2 orders of maganitude.



\noindent
\textbf{The Efficiency w.r.t. the Number, $|v_i.W|$, of Keywords Per Vertex:}
Figures~\ref{subfig:uni-w}, ~\ref{subfig:gau-w}, and ~\ref{subfig:zipf-w} report the effect of the number, $|v_i.W|$, of keywords per vertex on the S$^3$AND query performance, where $|v_i.W|$ varies from 1 to 5, and default values are used for other parameters. 
With more keywords in $v_i.W$ per vertex $v_i$, the pruning powers of \textit{keyword set} and \textit{ND lower bound pruning} become lower (i.e., with more candidate vertices), which thus leads to higher time cost. Nonetheless, the S$^3$AND query cost remains low (i.e., $0.05 \sim 31.03 \: sec$)  for different $|v_i.W|$ values, and outperforms $Baseline$ by 1-3 orders of magnitude.

\balance

\noindent
\textbf{The Efficiency w.r.t. the Size, $|\sum|$, of the Keyword Domain $\sum$:}
Figures~\ref{subfig:uni-domain}, ~\ref{subfig:gau-domain}, and ~\ref{subfig:zipf-domain} present the S$^3$AND query performance, by setting $|\sum|$ = 10, 30, 50, and 80, where other parameters are set to default values. 
With the same number, $|v_i.W|$, of keywords per vertex, higher $|\sum|$ value incurs more scattered keywords in the keyword domain, and leads to higher pruning power of keyword set pruning, resulting in fewer candidate vertices.
Therefore, as confirmed by figures, for larger $|\sum|$ value, the S$^3$AND query cost decreases and remains low (i.e., $0.98 \sim 23.23 \: sec$), outperforming $Baseline$ by 2-3 orders of magnitude.

\noindent
\textbf{The Efficiency w.r.t. the Size, $|V(q)|$, of Query Graph $q$:}
Figures~\ref{subfig:uni-vq}, ~\ref{subfig:gau-vq}, and ~\ref{subfig:zipf-vq} demonstrate the S$^3$AND query performance for different query graph sizes $|V(q)|$, where $|V(q)|$ = 3,5,8 and 10, and other parameters are set to their default values.
When the query graph size, $|V(q)|$, becomes larger, more sets of candidate vertices w.r.t. query vertices need to be retrieved and refined, resulting in higher query costs. Nevertheless, the time cost of our S$^3$AND approach still remains low (i.e., $0.82 \sim 174.63 \: sec$), which outperforms $Baseline$ by 1-2 orders of magnitude.


\noindent
\textbf{The Efficiency w.r.t. the Size, $|V(G)|$, of Data Graph $G$:}
Figures~\ref{subfig:uni-vG}, ~\ref{subfig:gau-vG}, and ~\ref{subfig:zipf-vG} test the scalability of our S$^3$AND approach for different data graph sizes $|V(G)|$ varying from $10 K$ to $30 M$, where default values are used for other parameters. From figures, we can see that, with the increase of the data graph size $|V(G)|$, the number of candidate vertices also increases, which leads to higher retrieval/refinement costs and in turn larger query time. 
For large-scale graphs with $30 M$ vertices, the time costs are less than $1,894.26 \: sec$ for all the three synthetic graphs, outperforming the $Baseline$ method by 2-3 orders of magnitude, which confirms the efficiency and scalability of our proposed S$^3$AND approach.

\subsection{Evaluation of the S$^3$AND Offline Pre-Computations}
\label{sec:offline_}
Figure~\ref{fig:offline time} presents the S$^3$AND offline pre-computation cost (including time costs of auxiliary data pre-computation and index construction) over real-world/synthetic graphs, where parameters are set to default values.
In Figure~\ref{subfig:S3AND_real_ofline-online}, for real-world graph size from $4K$ to $2.95M$, the overall offline pre-computation time varies from $33.03 \: sec$ to $11.85 \: h$. On the other hand, for synthetic graphs, when the graph size $|V(G)|$ is $50K$, the overall offline pre-computation time in Figure~\ref{subfig:S3AND_synthetic_ofline-online} varies from $44.47 \: sec$ to $46.15 \:sec$.

\begin{figure}[!t]
    \centering
    \subfigure[real-world graphs]{
            \includegraphics[width=0.48\linewidth]{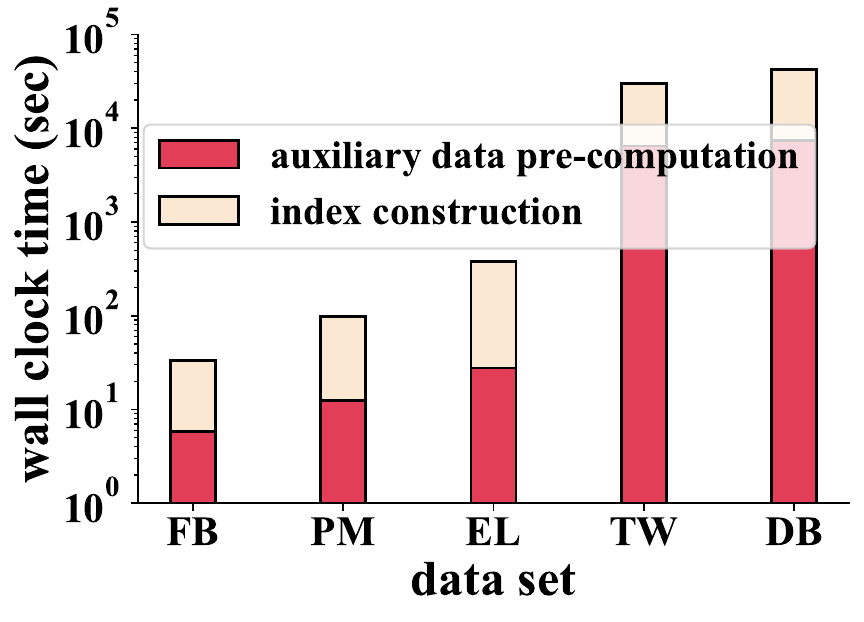}
            \label{subfig:S3AND_real_ofline-online}
    } \hspace{-0.2cm}
    \subfigure[synthetic graphs]{
        \includegraphics[width=0.48\linewidth]{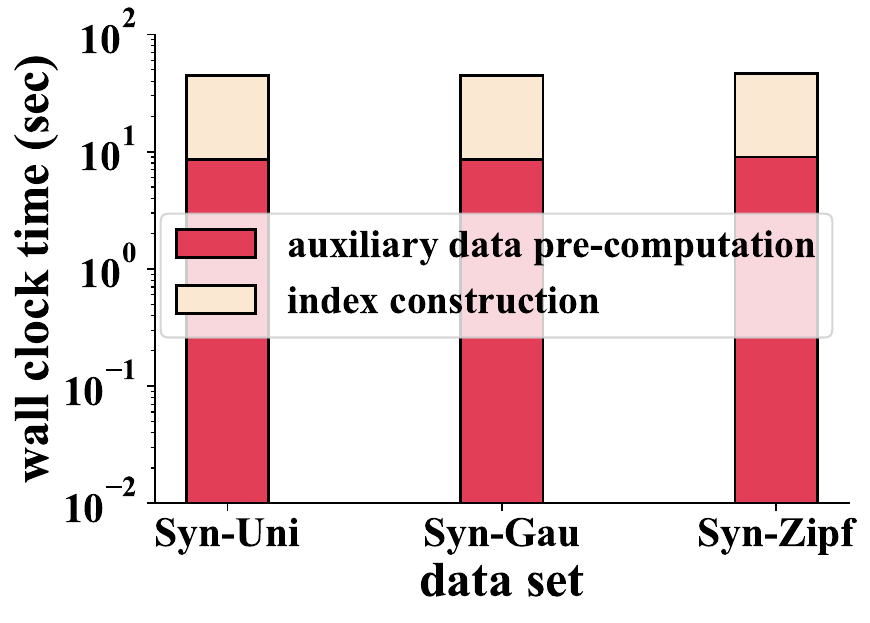}
        \label{subfig:S3AND_synthetic_ofline-online}
    }
    \vspace{-3ex}
    \caption{The time costs of offline auxiliary data pre-computation and index construction.}
    \label{fig:offline time}
\end{figure}

\begin{figure}[t]
    \centering
    \subfigure[real-world graphs]{
            \includegraphics[width=0.48\linewidth]{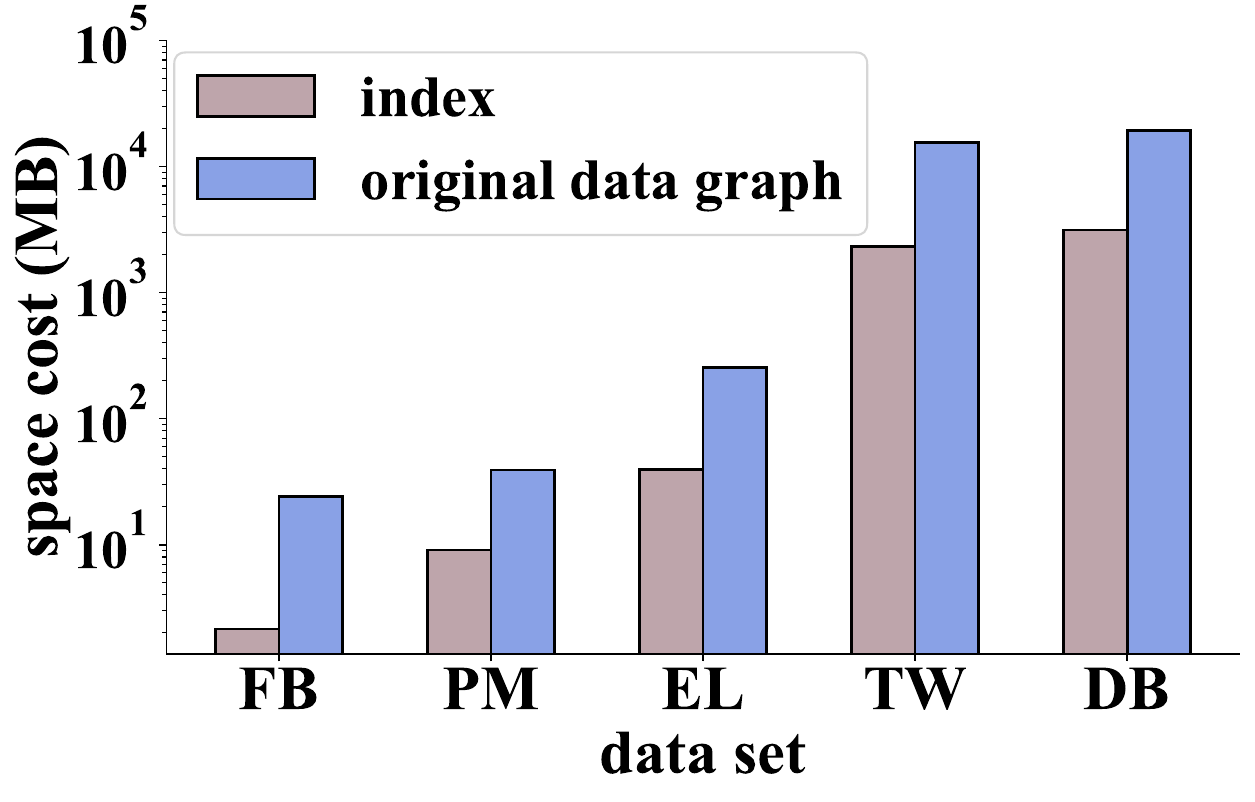}
            \label{subfig:S3AND_index_space_real}
    } \hspace{-0.2cm}
    \subfigure[synthetic graphs]{
        \includegraphics[width=0.48\linewidth]{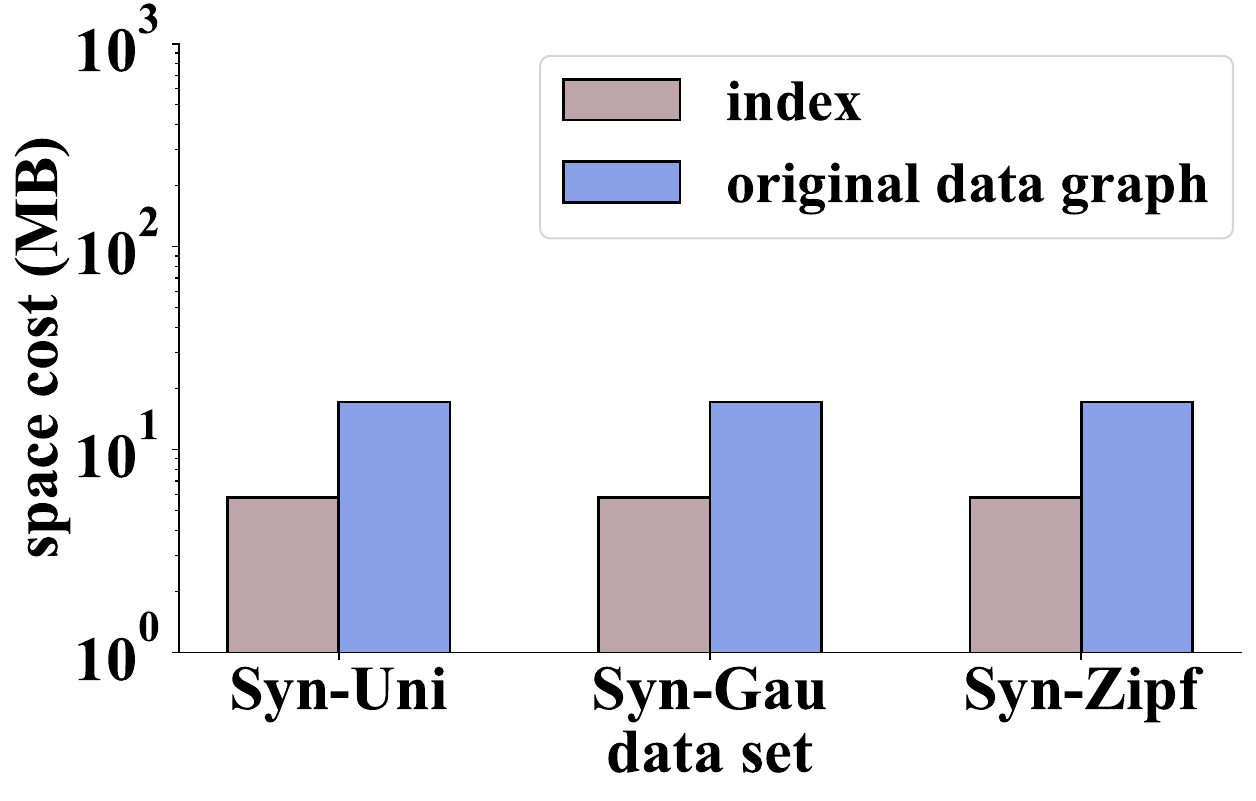}
        \label{subfig:S3AND_index_space_syn}
    }\vspace{-2ex}
    \caption{The space cost of the precomputed index.}
    \label{fig:space}
\end{figure}

Figures~\ref{subfig:S3AND_index_space_real} and \ref{subfig:S3AND_index_space_syn} show the statistics of the space consumption for the precomputed indexes over both real and synthetic graphs. From figures, we can see that for most real/synthetic graphs, the space cost of the index for our S$^3$AND algorithm is about one order of magnitude less than that of the original data graph.


\subsection{Case Study}
\label{subsec:case study}

\begin{figure}[!hb]
    \centering
    \setcounter{subfigure}{0}
    \subfigure[query graph $q$ ($|V(q)| = 5$)]{
    \includegraphics[width=0.475\linewidth]{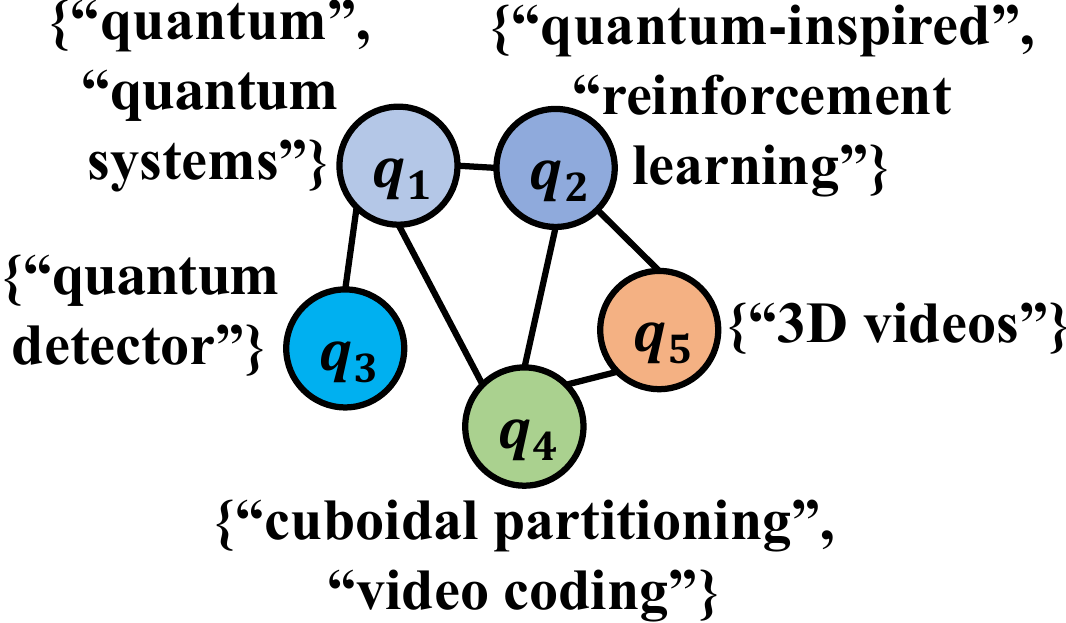}
    \label{subfig:caseq}
    }
    \subfigure[S$^3$AND ($AND(q,g)$$=$$0$; $f$$=$$MAX$)]{
    \includegraphics[width=0.475\linewidth]{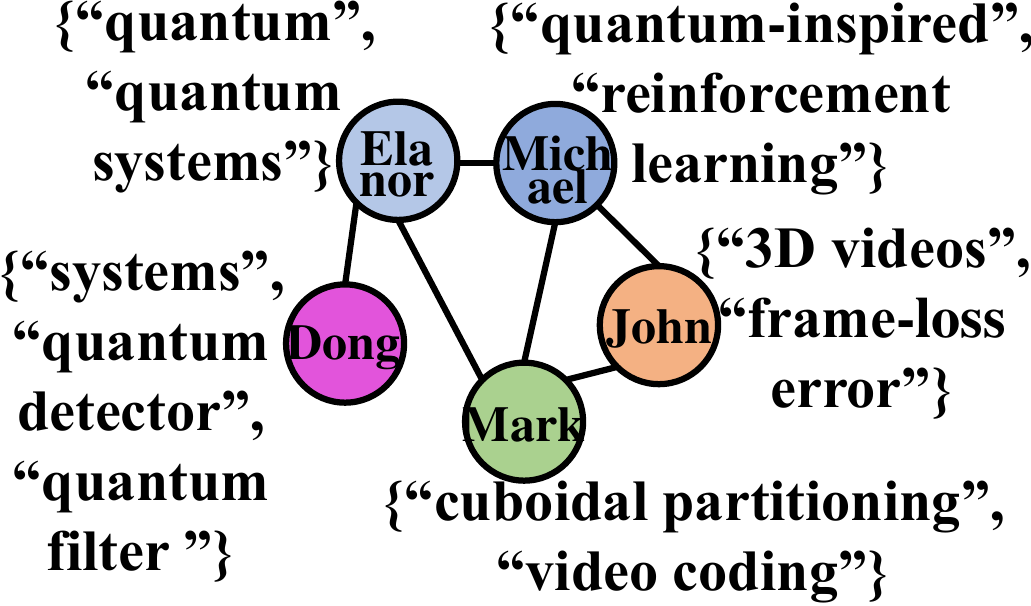}
    \label{subfig:caseS3}
    }\\
    \subfigure[$CSI\_GED$ ($GED(q, g) = 1$)]{
    \includegraphics[width=0.475\linewidth]{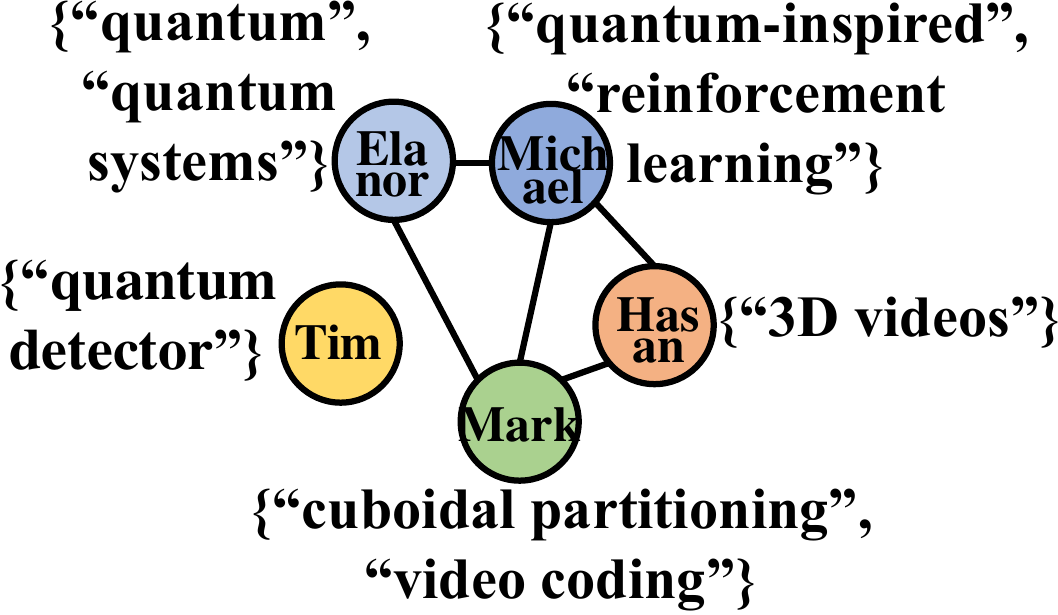}
    \label{subfig:GED}
    }
    \subfigure[$MCSPLIT$ ($|V(q)|$$-$$MCS(q,g)$$=1$)]{
    \includegraphics[width=0.475\linewidth]{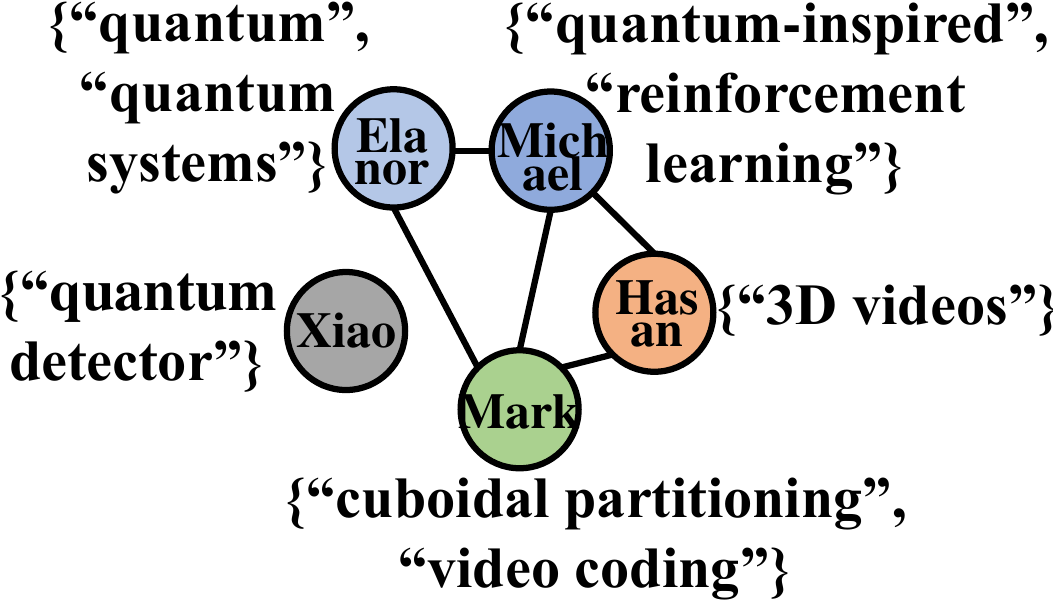}
    \label{subfig:MCS}
    }
    \caption{A case study of S$^3$AND, $CSI\_GED$, and $MCSPLIT$ over $DB$ graph data with query graph size 5.}
    \label{fig:case}
\end{figure}

In order to illustrate the effectiveness of our S$^3$AND semantics, in this subsection, we conduct a case study in Figure~\ref{fig:case} and evaluate top-1 query answer of our S$^3$AND semantics (with the smallest AND score), compared with that of baselines $CSI\_GED$ and $MCSPLIT$ (with the minimum GED or maximum MCS, respectively), over the $DBLPv14$ graph, $DB$, where the query graph size $|V(q)|$ is set to 5.

In this case study, a user may want to conduct new multidisciplinary research, especially in quantum and 3D vision areas. Thus, he/she can search for a collaboration team from the DBLP graph, whose members have quantum and/or 3D vision related background and have co-authored some papers before. Figure~\ref{subfig:caseq} shows the targeted collaboration team (i.e., query graph $q$), whereas Figure~\ref{subfig:caseS3} provides a subgraph answer satisfying the S$^3$AND constraints (under MAX aggregate). In particular, the author, ``Dong'', is an expert in the systems and quantum (i.e., matching with query vertex ``$q_3$''), and co-authored with ``Elanor'' before (matching with ``$q_1$''). Other authors, $Michael$, $Mark$, and $John$, have the expertise (keywords) related to the quantum and 3D videos, which include all the query keywords in each of query vertices $q_2$, $q_4$, and $q_5$, respectively.

In contrast, Figures~\ref{subfig:GED} and ~\ref{subfig:MCS} return top-1 query result of $CSI\_GED$ and $MCSPLIT$, respectively. However, these returned subgraph answers contain isolated vertices (e.g., ``Tim'' in Figure~\ref{subfig:GED} and ``Xiao'' in Figure~\ref{subfig:MCS}), which are not the desired collaboration teams (as some authors did not have co-author relationships with other members, and may incur high communication or technical cross-learning costs). Therefore, our S$^3$AND query is more effective to return subgraph answers that satisfy both keyword and neighbor difference conditions in such an application scenario/case.



\subsection{Parameter Tuning}
\label{sec:tuning}
In this subsection, we vary values of parameters such as $n$, $global\_iter$, and $local\_iter$, and evaluate/discuss how to tune these parameters. Moreover, we also discuss how to tune/choose the AND score threshold $\sigma_{MAX}$ or $\sigma_{SUM}$.

\underline{The Number, $n$, of Partitions:} The number, $n$, of partitions is one of inputs in Algorithm \ref{alg:cm}, which is invoked by line 5 of Algorithm \ref{alg:index2}. The number, $n$, of partitions in Algorithm \ref{alg:cm} is related to the fanout, $fanout$, of the index node in Algorithm \ref{alg:index2}, which is defined as the total node space (page size) divided by the space cost of each entry in the index node. Thus, in our experiments, we set $n$ to 16 ($= 4KB / 256 \: bytes$) by default, where $4KB$ is the space cost of a node (page), and $256 \: bytes$ is the space cost of an index entry (i.e., space costs of keyword bit vectors, neighbor keyword bit vectors, and the maximum number of distinct neighbor keywords).

We also test the effect of different $n$ values (i.e., 2, 8, 16, 24, and 32) on our S$^3$AND query performance in Figure~\ref{ngl:n}. From experimental results, we can see that as $n$ increases, the wall clock time slightly decreases first and then increases. This is because, larger $n$ values will lead to more branches in the tree index, possibly with smaller height, which incurs higher pruning power on branches and lower cost to traverse from root to leaf nodes. On the other hand, larger fanout $n$ will also increase the computation cost of searching within each index node. Nevertheless, from the experimental results, we can see that the wall clock time is not very sensitive to $n$. In our experiments, we simply set $n$ to 16.

\underline{The Number, $global\_iter$, of Global Iterations:} During the graph partitioning, we run multiple (i.e., $global\_iter$) global iterations with random starts of center vertices, in order to prevent our algorithm from falling into local optimality. Small $global\_iter$ value may lead to low partitioning quality (with local optimality), whereas large $global\_iter$ value may incur high time cost. 

Figure~\ref{ngl:g} illustrates the effect of $global\_iter$ on the S$^3$AND query cost over Facebook graph data, where $global\_iter$ varies from 5 to 20. We can see that the wall clock time is not very sensitive with respect to different $global\_iter$ values, which indicates that setting $global\_iter=5$ by default is sufficient for the index construction to facilitate our S$^3$AND approach.

\underline{The Number, $local\_iter$, of Local Iterations:} Each local iteration updates center vertices and performs re-assignment of vertices to partitions. Thus, larger $local\_iter$ values may achieve better partitioning strategies with higher quality, however, incur higher offline computation cost. 

Figure~\ref{ngl:l} varies parameter $local\_iter$ from 20 to 100 for the index construction over Facebook graph. Similar to previous experimental results, the wall clock time of our S$^3$AND approach over the resulting index is not very sensitive to $local\_iter$ values. Therefore, in our experiments, we set $local\_iter$ to 20 by default.

\underline{The Tuning of Threshold Parameters $\sigma_{MAX}$ and $\sigma_{SUM}$:} 
We also conduct a set of experiments on the frequency distributions of the AND scores over $FB$ and $PM$ graph data sets, for tuning threshold parameters, $\sigma_{MAX}$ and $\sigma_{SUM}$ (w.r.t, $MAX$ and $SUM$ aggregates, respectively).
Specifically, Figure~\ref{fig:ANDscore} presents the frequency distributions of the AND scores in $FB$ and $PM$ graphs, for the AND scores from 0 to 3 under $MAX$ aggregate and from 0 to 6 under $SUM$ aggregate, where other parameters are set to default values. From figures, we can see that, for both $FB$ and $PM$ graphs, most AND scores under $MAX$ aggregate have 0 or 1 frequency in Figure \ref{subfig:ANDscore_MAX}, whereas that with $SUM$ aggregate are distributed between 0 and 2 in Figure \ref{subfig:ANDscore_SUM}. Therefore, in order to tune the threshold parameters (e.g., $\sigma_{MAX}$ or $\sigma_{SUM}$ for a new graph data set), we can collect such statistics (i.e., the AND score histogram) with query graphs from historical logs, and set appropriate threshold $\sigma_{MAX}$ or $\sigma_{SUM}$ based on a user-specified query selectivity (i.e., the percentage/number of answer subgraphs that the user wants to obtain).

To summarize, our proposed S$^3$AND approach can achieve high pruning power and low wall clock time (compared with the baseline method, $Baseline$).

\begin{figure}[t]
    \centering
    \subfigure[$n$]{
        \includegraphics[width=0.308\linewidth]{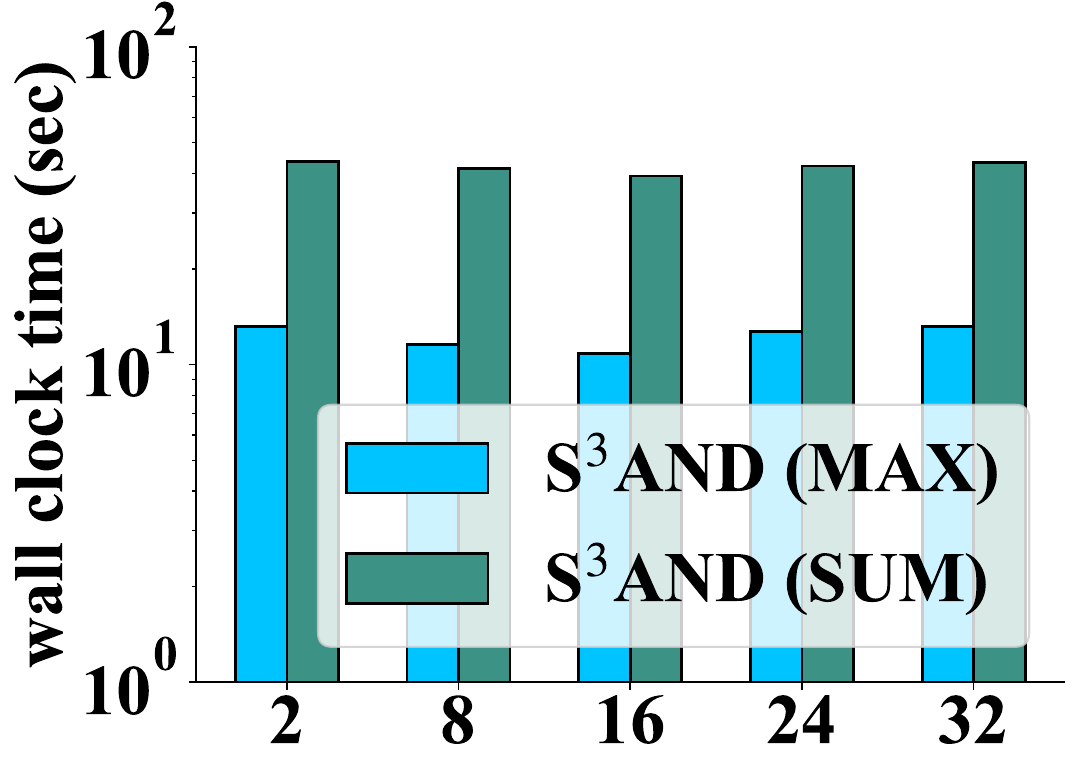}
        \label{ngl:n}
    }
    \subfigure[$global\_iter$]{
        \includegraphics[width=0.308\linewidth]{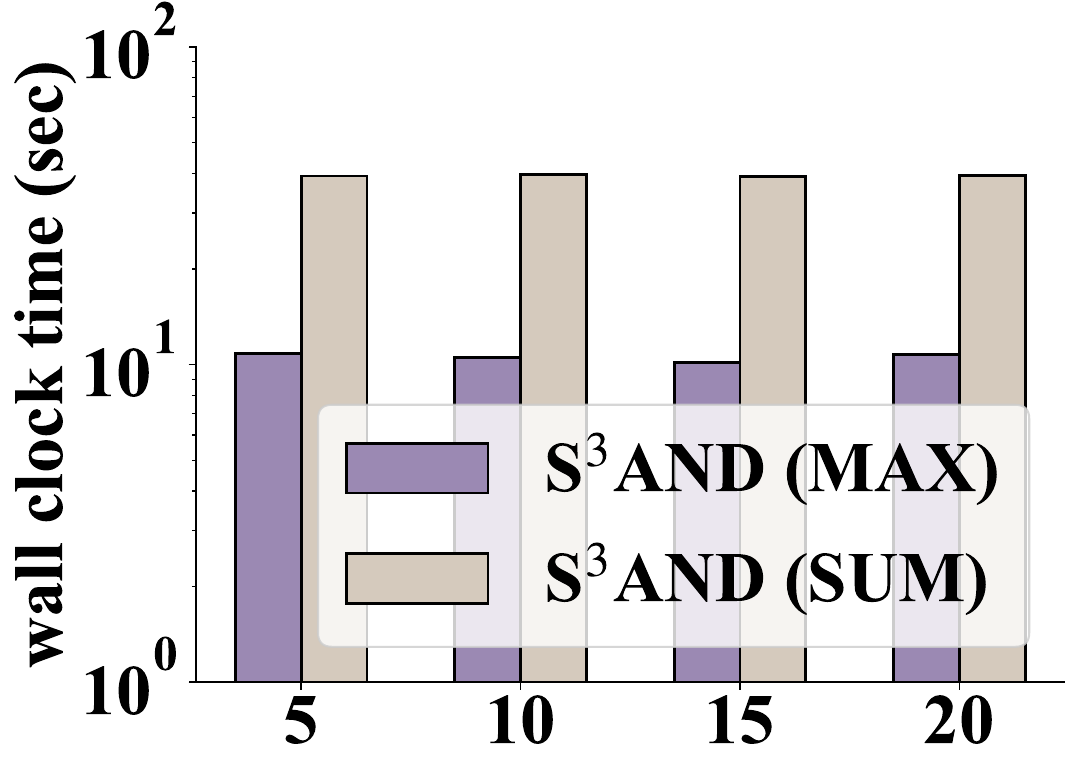}
        \label{ngl:g}
    }
    \subfigure[$local\_iter$]{
        \includegraphics[width=0.308\linewidth]{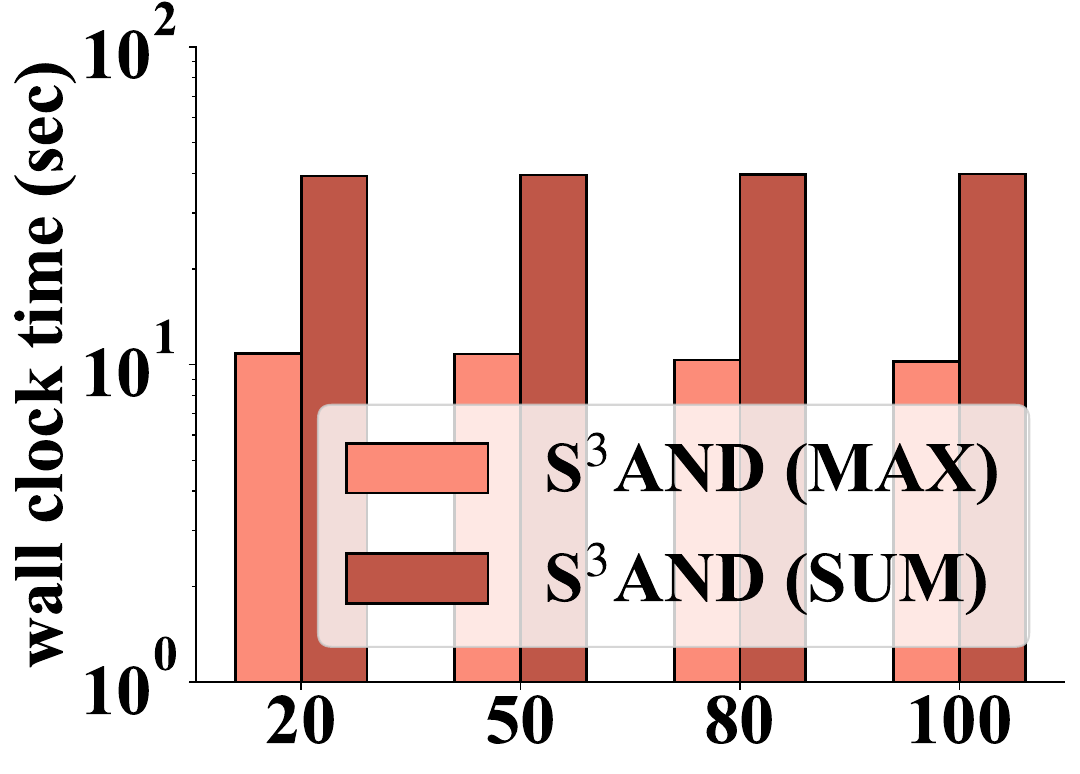}
        \label{ngl:l}
    }
    \caption{The S$^3$AND query efficiency on Facebook graph data set vs. $n$, $global\_iter$, and $local\_iter$.}
    \label{fig:ngl}
\end{figure}

\begin{figure}[H]
    \centering
    \subfigure[$AND$ score with $MAX$ aggregate]{
    \includegraphics[width=0.45\linewidth]{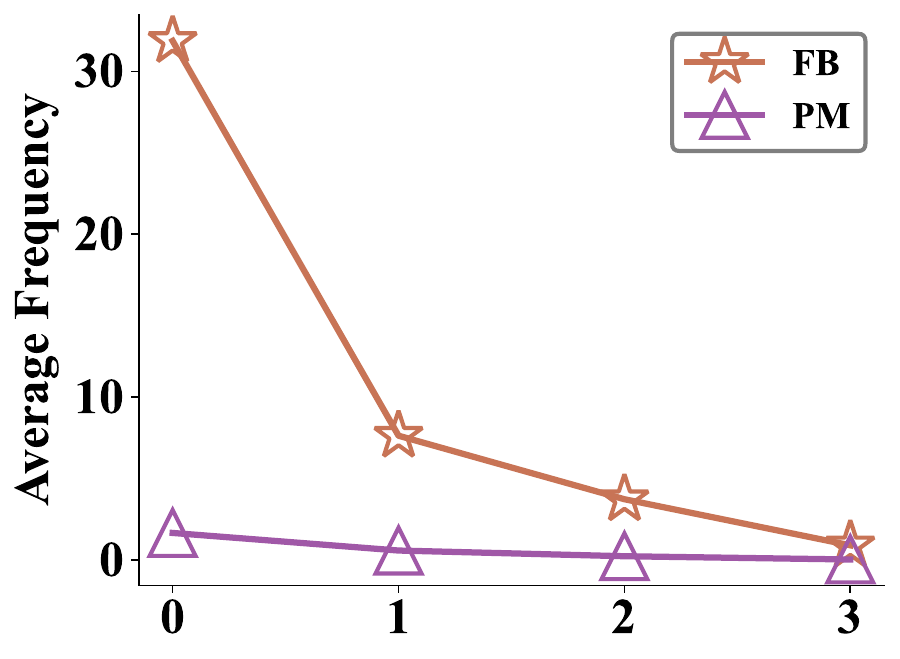}
    \label{subfig:ANDscore_MAX}
    }
    \subfigure[$AND$ score with $SUM$ aggregate]{
    \includegraphics[width=0.45\linewidth]{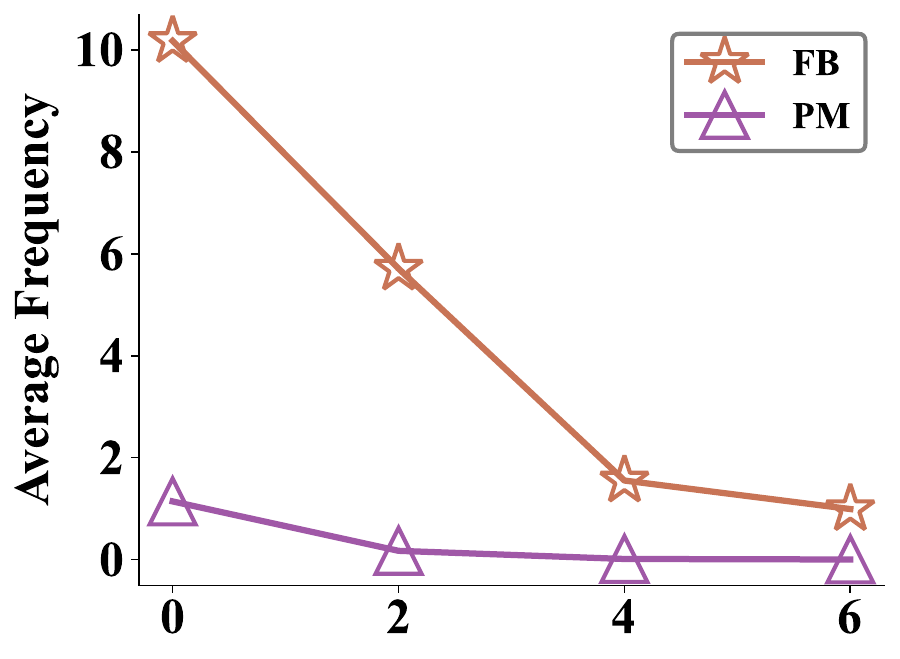}
    \label{subfig:ANDscore_SUM}
    }
    \vspace{-2ex}
    \caption{The average frequency distributions of the AND scores over FB and PM graph data sets.}
    \label{fig:ANDscore}
\end{figure}

\section{Related Work}
\label{sec:related_work}

In this section, we discuss previous works on subgraph matching and subgraph similarity search. 

{ \color{black}
\subsection{Subgraph Matching} 
\label{sec:sm}
Subgraph Matching (SM) is a fundamental task of graph mining that aims to discover important substructures over data graph \cite{YeLC24, ZhangLZL24, AraiFO23, Sun023, DBLP:conf/icde/YuanMZW22, DBLP:conf/icde/JinY0YQP21}. Recently, SM has been widely studied in community detection \cite{DBLP:conf/kdd/WuXZJSSZY22,berahmand2018community}, social network analysis \cite{DBLP:journals/vldb/HuangHZGYZ22, han2019efficient}, anomaly detection \cite{DBLP:conf/kdd/MaWYS23,DBLP:journals/tvcg/GuoJCGZC22}, etc. 
Previous works on SM can be divided into exact subgraph matching \cite{ammar2018distributed, lai2019distributed, sun2020rapidmatch} and approximate subgraph matching \cite{bhattarai2019ceci, han2019efficient} based on different classification standards.

\noindent{\bf Exact Subgraph Matching:} Existing works on exact subgraph matching considered backtracking-search-based~\cite{DBLP:conf/sigmod/BiCLQZ16, DBLP:conf/sigmod/HanKGPH19, DBLP:conf/sigmod/BhattaraiLH19} and multi-way-join-based algorithms~\cite{DBLP:journals/pvldb/LaiQYJLWHLQZZQZ19, sun2020rapidmatch, lai2019distributed}. 
The backtracking-search-based algorithm performs deep matching of the given query graph by vertex-to-vertex mapping and backtracks when the state matching fails. 
In particular, CECI~\cite{DBLP:conf/sigmod/BhattaraiLH19}, CFLMatch~\cite{DBLP:conf/sigmod/BiCLQZ16}, DP-iso~\cite{DBLP:conf/sigmod/HanKGPH19} optimize the overhead of generation of intermediate results by using preprocessing enumeration paradigms to execute a query. 
Examples of the multi-way-join-based algorithm include GpSM~\cite{tran2015fast}, which merges candidate edges based on pairwise join to obtain matching results and is suitable for tree-shaped or non-cyclic graph queries, and Graphflow~\cite{kankanamge2017graphflow}, which prunes neighbor nodes of candidate nodes based on worst-case optimal join to obtain the matching results and is suitable for dense cyclic graph queries.
Recently, embedding or learning-based approaches such as GNN-PE~\cite{ye2024efficient} considered the classic exact subgraph matching problem under the graph isomorphism semantics. GNN-PE~\cite{ye2024efficient} employed path embeddings for the exact subgraph matching problem over a data graph, where path embeddings are defined as the concatenation of embedding vectors from vertices on the path (i.e., embeddings of these vertices and their 1-hop neighbors), learned by \textit{Graph Neural Networks} (GNNs). In GNN-PE, it is assumed that each vertex in the data graph is only associated with a single keyword (rather than a keyword set in S$^3$AND), and the subgraph matching considers the graph isomorphism (instead of S$^3$AND matching semantics such as keyword set containment and aggregated neighbor difference constraints). Therefore, with a different graph data model and query semantics, we cannot directly use previous techniques in GNN-PE for tackling our S$^3$AND problem.

\vspace{0.5ex}\noindent{\bf Approximate Subgraph Matching:} When the response time is much more important than the accuracy, approximate subgraph matching improves the efficiency of subgraph matching by returning top-$k$ approximate subgraphs that are similar to the query graph, and is widely used in real applications~\cite{DBLP:conf/kdd/MaWYS23,DBLP:journals/tvcg/GuoJCGZC22,han2019efficient}. Existing approximate subgraph matching algorithms usually searched for top-$k$ similar subgraphs from a (large) data graph, by setting different similarity metrics for various scenarios, e.g., GED~\cite{gouda2016csi_ged, ibragimov2013gedevo, zeng2009comparing} and GBD~\cite{li2018efficient}.  Although these matches can give answers quickly, they do not ensure the accuracy of the returned subgraph answers and are more limited to the task scenario (e.g., the algorithms cannot give the exact locations of similarity subgraphs in the data graph).


\subsection{Subgraph Similarity Search}
\label{sec:sss}

Previous works on subgraph similarity search have conducted extensive research on subgraph partitioning~\cite{zhao2013partition, liang2017similarity}, filtering optimization~\cite{chen2019efficient, yang2014schemaless}, and indexing retrieval~\cite{wang2012efficient, wang2010efficiently, khan2013nema} to improve the efficiency. 
NeMa~\cite{khan2013nema} obtained top-$k$ subgraphs with the minimum matching costs, defined as the sum of \textit{keyword matching} and \textit{distance proximity costs} between query and data vertices. Here, the \textit{keyword matching cost} is given by the \textit{Jaccard similarity} over keyword sets from a pair of query and data vertices. Moreover, the \textit{distance proximity cost} is defined as the difference between \textit{neighborhood vectors}~\cite{khan2013nema} from a pair of query and data vertices, where the \textit{neighborhood vector} contains the distances from the query/data vertex to its neighbors within $h$-hop away from the vertex. In contrast, our S$^3$AND query semantics consider the containment relationship of keyword sets for the vertex matching (i.e., different from the \textit{Jaccard similarity measure} in NeMa), and take into account the structural difference between subgraph $g$ and query graph $q$ (i.e., aggregated 1-hop neighbor difference of each vertex $v_i$, compared with query vertex $q_j$) which differs from the NeMa semantics (i.e., the distance proximity cost, caring more about the similarity of distances from query/data vertices to their $h$-hop neighbors). Thus, due to distinct query semantics, we cannot directly borrow the techniques proposed for NeMa to solve our S$^3$AND problem.
SLQ~\cite{yang2014schemaless} obtains the top-$k$ subgraphs with the highest ranking scores, given by the sum of \textit{edge} and \textit{node matching costs}, where the \textit{edge matching cost} (or \textit{node matching cost}) is defined as the (weighted) number of transformation functions (pre-defined in a library) that can transform the data edge (or data node) to the query edge (or query node). Different from SLQ that considered the graph data model with semantic information in vertices and edges, the graph model in our S$^3$AND problem assumes vertices associated with keyword sets. Furthermore, our S$^3$AND query semantics focus on the 1-hop neighbor structural difference between subgraph $g$ and query graph $q$, which differs from the \textit{ranking scores} in SLQ. Thus, with different graph data model and query semantics, we cannot directly apply the approaches proposed in SLQ to tackle our S$^3$AND problem.
}
Recently, with the development of neural networks, e.g., GNN and GCN, more and more embedding-based subgraph similarity search algorithms~\cite{bause2022embassi, qin2020ghashing, YeLC24, bai2019simgnn, li2019graph} have been proposed, which can achieve faster online processing time.
However, the accuracy and model training cost of these methods are still insufficient for the needs of critical applications.
Due to different graph similarity semantics, we cannot directly borrow previous works on subgraph similarity search to solve our S$^3$AND problem.


\section{Conclusions}
\label{sec:conclusions}

In this paper, we formulate a novel problem, \textit{subgraph similarity search under aggregated neighbor difference semantics} (S$^3$AND), which has broad applications (e.g., collaborative team detection and fraud syndicate identification) in real-world scenarios.
To enable efficient online S$^3$AND queries, we propose two pruning strategies (i.e., \textit{keyword set} and \textit{AND lower bound pruning}), to filter out false alarms of candidate vertices/subgraphs. We also devise a tree index on offline pre-computed data, which can help apply our proposed pruning strategies to retrieve candidate subgraphs during the index traversal. Finally, we conduct extensive experiments to confirm the effectiveness and efficiency of our proposed S$^3$AND approach on real and synthetic graphs.


\clearpage

\bibliographystyle{ACM-Reference-Format}

\bibliography{dblp}

\end{document}